\documentclass[11pt, hidelinks]{article} 

\usepackage{hyperref}
\hypersetup{colorlinks,linkcolor={blue},citecolor={blue},urlcolor={red}}

\usepackage{setspace}
\usepackage{url}
\usepackage{tcolorbox}
\usepackage{fullpage}
\usepackage{algorithmic}
\usepackage{color,xcolor}

\usepackage{enumitem}

\usepackage[mathscr]{eucal}
\DeclareMathAlphabet{\mathpzc}{OT1}{pzc}{m}{it}

\usepackage{amsbsy}

\usepackage[linesnumbered,ruled]{algorithm2e}

\usepackage{fancyhdr}

\usepackage{microtype}

\usepackage{amsmath,amssymb,amsthm, amsfonts}

\usepackage{enumitem}

\usepackage{comment}

\usepackage{xifthen}

\usepackage[noabbrev,nameinlink]{cleveref}

\DontPrintSemicolon

\usepackage{basic_latex_macros}

\usepackage{flora_macros}

\theoremstyle{plain}
\newtheorem{theorem}{Theorem}
\newtheorem{lemma}{Lemma} \newtheorem{corollary}{Corollary}
\newtheorem{proposition}{Proposition}

\theoremstyle{definition}

\usepackage[giveninits=true, 
  maxnames = 2, 
  backref=true,
  backend=bibtex,
  style=alphabetic, sortlocale=de_DE, natbib=true, doi=false,
  isbn=false, url=false]{biblatex}

\newcommand{\ContractMenu}{\ensuremath{\mathscr{M}}}

\newcommand{\typeNum}{\ensuremath{N}}
\newcommand{\numType}{\typeNum}

\newcommand{\SelFunPlain}{\ensuremath{\varphi}}

\newcommand{\SelFun}{\ensuremath{\SelFunPlain_{\ContractMenu}}}
\newcommand{\reportSupp}{\ensuremath{\mathcal{P}}}

\newcommand{\figdir}{./Figure}

\newcommand{\tdr}{\ensuremath{\operatorname{TDR}}}

\newcommand{\Gfun}{\ensuremath{\mathcal{G}}}

\newcommand{\gstar}[1]{\ensuremath{g^*_{#1}}}
\newcommand{\Score}{\ensuremath{\mathcal{S}}}

\newcommand{\eps}[1]{\ensuremath{\specfunc{(#1)}}}

\newcommand{\Interval}{\ensuremath{\mathcal{I}}}
\newcommand{\intleft}{\ensuremath{\ell}}
\newcommand{\intright}{\ensuremath{r}}

\newcommand{\specfunc}{\ensuremath{\epsilon}}

\newcommand{\PsiBase}{\ensuremath{\smallsub{\Psi}{\mbox{base}}}}

\newcommand{\TypeDist}{\typeDist}

\newcommand{\BayesRisk}{\ensuremath{\operatorname{BR}}}

\newcommand{\mygraybox}[1]{
{\begin{tcolorbox}[
    colback=gray!10,       
    colframe=gray!70,      
    arc=6pt,               
    boxrule=0.8pt,         
    left=8pt,              
    right=8pt,             
    top=6pt,               
    bottom=6pt,             
    parbox=false
]
#1
  \end{tcolorbox}
  }
}

\begin{document}

\begin{center}

  {\bf \Large Instance-Adaptive Hypothesis Tests with Heterogeneous Agents} \\

   \vspace*{.2in}
  \begin{tabular}{ccc}
    Flora C. Shi$^\dagger$ & Martin J. Wainwright$^{\dagger, \star}$ &
    Stephen Bates$^\dagger$
  \end{tabular}
  
  \vspace*{0.1in}
  
  \begin{tabular}{c}
    Laboratory for Information and Decision Systems \\
    Statistics and Data Science Center \\
    EECS$^{\dagger}$ and Mathematics$^{\star}$ \\
    Massachusetts Institute of Technology
  \end{tabular}
  
  \vspace*{.2in}

  \begin{abstract}

  We study hypothesis testing over a heterogeneous population of
  strategic agents with private information. Any single test applied
  uniformly across the population yields statistical error that is
  sub-optimal relative to the performance of an oracle given access to
  the private information.  We show how it is possible to design menus
  of statistical contracts that pair type-optimal tests with payoff
  structures, inducing agents to self-select according to their
  private information. This separating menu elicits agent types and
  enables the principal to match the oracle performance even without a
  priori knowledge of the agent type.  Our main result fully
  characterizes the collection of all separating menus that are
  instance-adaptive, matching oracle performance for an arbitrary
  population of heterogeneous agents. We identify designs where
  information elicitation is essentially costless, requiring
  negligible additional expense relative to a single-test benchmark,
  while improving statistical performance. Our work establishes a
  connection between proper scoring rules and menu design, showing how
  the structure of the hypothesis test constrains the elicitable
  information. Numerical examples illustrate the geometry of
  separating menus and the improvements they deliver in error
  trade-offs. Overall, our results connect statistical decision theory
  with mechanism design, demonstrating how heterogeneity and strategic
  participation can be harnessed to improve efficiency in hypothesis
  testing.
  \end{abstract}
\end{center}


\section{Introduction}
\label{SecIntro}

Hypothesis testing involves making a discrete choice and constitutes
the most basic form of decision-making under uncertainty. It occupies
a central role in practice, serving as an assessment of evidence in
most scientific research.  Testing, however, rarely occurs in
isolation.  In a broader system-level context, outcomes of a test have
downstream implications, such as deploying a new drug, adding a new
feature to a commercial product, or publishing a scientific finding.
While all parties involved are affected, stakeholders may be
differentially impacted by the results. In the other direction,
upstream of the phase of data collection, these same parties may have
additional, potentially private information about the object of study
and make choices that affect the data generation. For example, a
pharmaceutical company sponsoring a phase III clinical trial for a
candidate drug implicitly has a belief about the drug's efficacy and
chooses to run a trial accordingly.  In this way, the distribution of
the data observed, and thus the performance of any method for
statistical decision-making, is altered by strategic behavior of the
agents involved.

In this paper, we study hypothesis testing based on data collected
from a heterogeneous population of agents, and analyze how their
strategic behavior impacts the design of optimal statistical tests.
More specifically, we consider a \emph{principal} that seeks to carry
out a collection of hypothesis tests. Each test is controlled by an
\emph{agent} that possesses private information---not known to the
principal---and this private information can differ across agents.
Our main contribution is to show how to design a protocol for
hypothesis testing that yields statistically optimal results even when
the private information of the agents is \emph{not known} in advance.

As a simple illustration, suppose there are two agent types: a
``good'' type with a high prior probability of being non-null and a
``bad'' type with a low prior probability of being non-null. The
decision-maker evaluates agents using $p$-value thresholds, which
determine the stringency of the test: smaller thresholds reduce false
positives but also reduce true discoveries.  For a heterogeneous
collection of tests, one natural measure of aggregate performance is
based on the false discovery rate and true discovery rate; cf.
equations~\eqref{EqnDefnFDR} and~\eqref{EqnDefnETD}. The false
discovery rate (FDR) corresponds to the fraction of rejections that
are false, while the true discovery rate (TDR) corresponds to the rate of
true discoveries across the mixture of types, and captures a notion of
overall power.

\begin{figure}[htbp]
    \centering
    \includegraphics[width=0.6\textwidth]{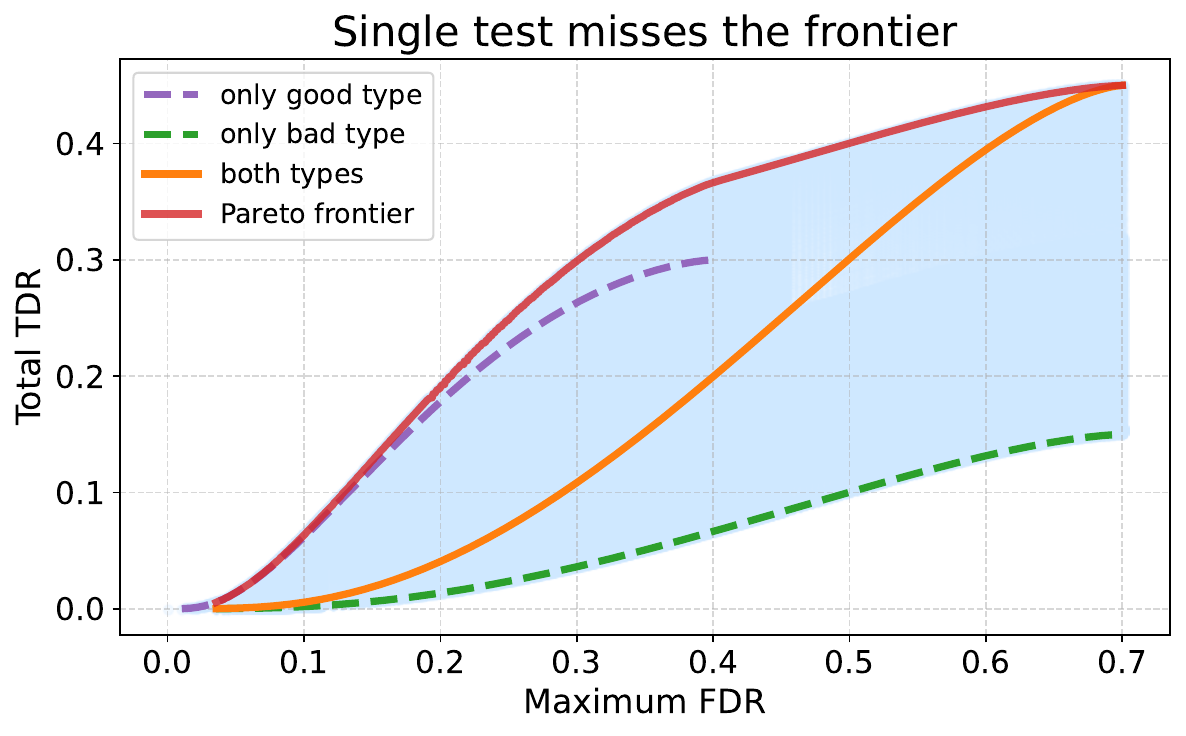}
    \caption{trade-offs between FDR and TDR for a population with two
      agent types: ``good'' and ``bad''. The light-blue region
      corresponds to all FDR-TDR pairs that can be achieved by any
      testing protocol. Its upper boundary (solid red) defines the
      Pareto frontier that can be achieved by an oracle that knows the
      agent type associated with each observation, and applies
      different test thresholds to each.  The orange solid curve shows
      the result of applying a single uniform test to the full agent
      population.  The dotted lines show the performance of applying a
      single test to only the ``good'' agents (purple line), and only
      the ``bad'' agents (green line); applying these two protocols
      also requires knowledge of the agent type.}
    \label{FigFrontier}
\end{figure}

\Cref{FigFrontier} illustrates how heterogeneity undermines
efficiency. The light-blue region shows all achievable (FDR, TDR)
pairs when the decision-maker is permitted to apply different tests to
each agent type; the upper boundary of this region (marked in a solid
red line) corresponds to the \emph{oracle Pareto frontier} that could
be achieved if the principal were given access to the hidden agent
type.  A single threshold applied uniformly across both types (orange
curve) lies substantially below the Pareto frontier.  We can also
compare to two other single threshold protocols that do require
knowledge of agent type.  Applying a test to only the ``good'' agents
yields the purple curve; it has a more favorable trade-off, but still
lies below the Pareto frontier.  Testing only the ``bad'' type (green
curve) defines the lower boundary of the achievable region.

When the agent types are known, it is possible to tailor thresholds to
different agent types so as to balance error trade-offs more
effectively: moderate thresholds for good type to increase power, and
stricter thresholds for bad types to limit false discoveries.  It is
exactly this tailored choice of thresholds that yields the oracle
Pareto frontier (red curve).  However, in practice, agent types are
\emph{not} directly observable, so we are led to the following question:

\mygraybox{{\bf{Question:}} \emph{Is it possible to match the
  statistical performance of the oracle even when agent types are
  unknown a priori?}}

The main contribution of this paper is to answer this question in a constructive way: we both characterize when it is possible to match the oracle and provide a concrete procedure for constructing testing protocols that do so. The key insight in our work is that, because agents are assumed to behave strategically, their actions can reveal private information. We leverage this by offering a menu of possible statistical tests, each paired with an associated payoff structure---together forming what we call a contract. We show that it is possible to design a menu of contracts such that an agent’s choice of contract reveals their private information, and moreover, such that each agent selects the statistical threshold that is optimal for their type. This design yields a testing protocol whose statistical error trade-offs match the oracle Pareto frontier, even though the principal has no \emph{a priori} knowledge of agent types.


\subsection{Our contributions}

We study a game-theoretic model of hypothesis testing with a
heterogeneous population of strategic agents. The decision-maker
(principal) designs a hypothesis testing protocol, while agents decide
whether to participate based on private information and expected
payoffs. Participation requires a fixed cost, and agents receive a
reward if they pass the test. We formalize this interaction through a
menu of statistical contracts, as defined precisely in
equation~\eqref{EqnContractMenu}, where each contract specifies the
testing protocol, the upfront cost, and the reward.

There is an evolving line of work on principal-agent
testing~\citep{tetenov2016economic, bates2022principalagent,
  bates2023incentivetheoretic, shi2024sharp, hossain2025,
  viviano2024modelmultiplehypothesistesting}, and within this context,
our analysis makes the following contributions:

\setlength{\leftmargini}{1em} 
\begin{itemize}
    \item {\bf Menus for heterogeneous agents.} Prior studies largely
      focus on single agents or uniform tests that control only false
      discoveries. We instead analyze heterogeneous populations and
      the optimal trade-off between the two competing types of
      statistical error. We show that the principal can implement a
      menu of statistical contracts to fully elicit agents’ private
      information and assign type-optimal thresholds, thereby matching
      the statistical performance of an oracle given access to agent
      types in advance.  Our main result (\Cref{th:psr_menu}) fully
      characterizes all separating menus that have this
      instance-adaptive optimality.
    \item {\bf Costs of elicitation.} A natural concern is whether
      eliciting private information requires substantial financial
      resources. We show that the principal can construct menus (see
      \Cref{co:Gfun_varyingR}, illustrated in \Cref{FigPSRVarying})
      that achieve statistical efficiency while incurring arbitrarily
      small financial cost relative to offering a single
      contract. This establishes that optimal error trade-offs can be
      attained essentially ``for free."
    \item {\bf Connection to scoring rules and test structure.}  An
      interesting by-product of our analysis is to establish a
      connection between menu design and strictly proper scoring rules
      (\Cref{SecMain:connection_to_psr}). However, in contrast to
      classical scoring rules, the principal cannot directly specify
      the agent’s utility but must induce it indirectly through
      contracts tied to hypothesis tests. When all contract parameters
      are flexible, full elicitation is possible. Under
      constraints---e.g., when rewards are fixed---the structure of
      the statistical testing problem itself dictates which types can
      be elicited and at what cost (see \Cref{co:Gfun_constR}).
\end{itemize}

Taken together, our results illustrate a key principle: heterogeneity
and strategic behavior, which are often viewed as obstacles, can
instead be leveraged as tools. By carefully designing testing
protocols and incentives, the principal can transform private
information and self-interested participation into a mechanism for
statistically optimal decision-making---and remarkably, this can be
achieved with minimal financial cost.

\subsection{Related work}

Our work connects to several strands of literature at the intersection
of contract theory, statistical decision-making, and strategic
behavior.  In addition, by comparing statistical performance to an
oracle, our work shares the spirit of statistical literature on
adaptive estimation and testing.

Let us begin by discussing the connections to screening and incentive
design in contract theory~\citep{laffont2001, bolton2004contract,
  salanie2005economics}. Here the focus of the study is how a principal
can elicit private information from heterogeneous agents by offering a
menu of contracts. Classical applications include nonlinear
pricing~\citep{mussa1978monopoly, maskin1984monopoly},
auctions~\citep{myerson1981optimal, maskin1984optimal}, and
regulation~\citep{baron1982regulating, laffont1993theory}.  Our work
shares this idea of elicitation, but shows that hypothesis tests
themselves, paired with suitable payoff structures, can act as
screening devices. Related recent work has used experimental designs
as screening mechanisms~\citep{wittbrodt2025delegating,
  yoder2022designing, wang2023contractingheterogeneousresearchers,
  jagadeesan2025publication}, underscoring the broader potential of
statistical objects for information elicitation.  Elicitation is
naturally connected to the notion of proper scoring
rules~\citep{brier1950verification, good1952rational,
  mccarthy1956measures, savage1971elicitation, gneiting2007strictly},
which correspond to utility functions designed to elicit truthful
reports.  Our mechanism also encourages truthful elicitation, but with
the important distinction that we have only indirect and partial
control over the utility via the specification of the hypothesis test.

More broadly, our work relates to a growing literature on strategic
behavior in decision-making, where agents may manipulate data,
analysis, or participation to influence outcomes. A prominent example
is $p$-hacking, where researchers adjust analyses to obtain
significant results, inflating false positives. The statistics
literature has responded with selective inference
methods~\citep{taylor2015statistical, berk2013valid}, while economic
analyses take a game-theoretic perspective. For instance,
\citet{mccloskey2024critical} designed critical values robust to
strategic manipulation, and \citet{jagadeesan2025publication} studied
how differences in private costs and incentives shape optimal
publication rules. \citet{spiess2018optimal} showed that restricting
researchers to fixed-bias estimators can improve inference when
planner's and researcher's objectives diverge.  Beyond $p$-hacking,
the machine learning community has explored mechanism design under
strategic manipulation of data, such as altering features or labels to
influence classifiers~\citep{hardt2016strategic, dong2018strategic} or
regressions \citep{dekel2010incentive, perote2004strategy,
  chen2018strategyproof}. In contrast, the strategic behavior in our
setting takes the form of participation decisions: agents cannot alter
test outcomes directly, but they can choose whether to engage. This
captures environments where experiments are preregistered or data
collection is externally verified, so participation is the primary
lever for influence.

As noted above, our work contributes directly to the evolving
literature on principal–agent hypothesis
testing. \citet{tetenov2016economic} provided an early analysis
linking type I error control to a cost–profit ratio. Subsequent
work~\citep{bates2022principalagent, bates2023incentivetheoretic,
  shi2024sharp, hossain2025} expanded the framework to incorporate
risk aversion, stochastic rewards, and simultaneous control of type I
and II errors. Our contribution departs from this line by focusing on
heterogeneous populations and showing how menus of contracts can
achieve type-optimal thresholds and optimal error trade-offs. In doing
so, we integrate heterogeneity, screening, and strategic participation
into a unified framework that ties statistical objectives to
incentive-compatible design.

Finally, we define optimality in this paper with reference to an
oracle that is given direct access to the agents' private information.
Our work thus shares the spirit of work on adaptive estimation and
testing (e.g.,~\citep{CaiLow2006, Lep90,Spo96}).  This classical work
considers tests or estimators that adapt to unknown problem structure
such as smoothness or sparsity, whereas we study adaptation to the
unknown private information of a collection of agents.

\paragraph{Paper organization:}

\Cref{SecSetup} formalizes the menu-based principal–agent testing
framework, specifying agents’ strategic behavior and the principal’s
statistical objectives. \Cref{SecMain} is devoted to the statement of
main results on separating menus (\Cref{SecMain:psr_menu}), along with
results on menu constructions that satisfy pre-specified criteria
(\Cref{SecMain:opt_menu}). Our central theorem (\Cref{th:psr_menu}),
stated in \Cref{SecMain:characterize}, characterizes all separating
menus that elicit private information and implement type-optimal
thresholds for arbitrary heterogeneous populations.
In~\Cref{SecMain:oracle}, we show how this theorem allows the
principal to match the statistical performance of the oracle
(\Cref{CorOracle}).  We highlight a connection to proper scoring rules
in \Cref{SecMain:connection_to_psr}, and analyze the financial costs
of elicitation in~\Cref{SecMain:cost_of_elicitation}.  In
\Cref{SecMain:opt_menu}, we study menu constructions under two
practical design criteria: one that minimizes the principal’s
financial cost, showing that information elicitation can be achieved
``for free," and another that addresses settings where the principal
can only partially specify contracts. We illustrate the implications
of these results through synthetic examples.  We conclude in
\Cref{SecDiscussion} with a discussion.  We examine the robustness of
menu-based testing to model misspecification
in~\Cref{Sec:sensitivity}.


\section{Mathematical Model}
\label{SecSetup}

In this section, we formalize the problem by specifying the structure
of the hypothesis tests, the principal–agent interaction, the agent’s
utility maximization problem, and the principal’s statistical decision
problem.

\subsection{Strategic hypothesis test}
\label{SecSetup:test_structure}

We consider a game-theoretic framework involving two parties: a
principal, who serves as a statistical regulator, and an agent. The
principal’s role is to approve or deny proposals submitted by
agents. To make these decisions, the principal offers the agent a menu
of statistical contracts at different costs (described below). Agents
observe the entire menu and choose whether to select one contract from
the menu or to opt out, based on their private information and
utility. Upon selecting a contract, data is collected and reported to
the principal, who conducts a hypothesis test according to the terms
specified in the chosen contract. We make all of these steps precise
below.


 \subsubsection*{Hypothesis space and agent type}
\label{SecSetup:hypothesis_space}

We assume that there exists a hidden random parameter $\theta$, which
takes two possible values $\theta \in \Theta = \{\theta_0, \theta_1\}$
and represents the latent quality of any proposal. The null value
$\theta_0$ corresponds to an ineffective proposal and the non-null
value $\theta_1$ corresponds to an effective proposal. Neither the
principal nor the agent knows the true value of $\theta$. However, the
agent has private knowledge of a prior distribution $\priorDist$ over
$\Theta$, which is inaccessible to the principal. While the principal
could attempt to ask the agent directly about $\priorDist$, there is
no guarantee the agent would report it truthfully. As we later show,
the principal can instead offer a menu of contracts designed to elicit
information about $\priorDist$.

We consider a setting with a population of agents, each characterized
by a potentially distinct prior distribution $\priorDist$. Since the
parameter space is binary, this distribution is fully specified by the
\emph{prior null probability} $\priorNull \defn \priorDist(\theta =
\theta_0)$.  We can classify agents according to the value of their
prior null: in particular, we say that an \emph{agent is of type
$\priorNull$} if he has prior null probability equal to $\priorNull$.
The overall population of agents consists of different types occurring
with a certain frequency, and we let $\typeDist$ denote a distribution
over possible types $\priorNull$. We do not assume that the principal
knows $\typeDist$.

\subsubsection*{The principal-agent interaction}

We now formalize the notion of a statistical contract and its role in
the principal-agent interaction. The principal aims to approve
effective proposals and deny ineffective ones by conducting hypothesis
tests. To this end, she designs a contract menu, meaning a collection
of contracts indexed by reported type $\reportNull \in [0,1]$, where
each contract specifies how the hypothesis test will be conducted. Any
\emph{contract menu} takes the form
\begin{align}
\label{EqnContractMenu}
\ContractMenu & \defn \Big \{ (\fthreshold{\reportNull},
\freward{\reportNull}, \fcost{\reportNull}) \mid \reportNull \in \reportSupp \Big
\},
\end{align}
where $\reportSupp \subseteq [0,1]$ is a subset of values associated
with contracts, and any contract is defined by a triplet with the following components:
\begin{itemize}[leftmargin=1em] 
\item a threshold $\fthreshold{\reportNull} \in [0,1] $ used in the
  principal’s binary hypothesis test, as described below.
\item the reward $\freward{\reportNull}$ received by the agent if the
  proposal is approved, and
\item the cost $\fcost{\reportNull}$ paid by the agent to opt in
  (and hence to generate evidence).
\end{itemize}
\noindent The agent observes the full contract
menu~\eqref{EqnContractMenu} before making any selection. Based on his private information $\priorNull$, the agent decides whether to opt in or opt out. If the
agent opts in, he selects a type $\reportNull \in
\reportSupp$ to the principal. The contract corresponding to this
reported type $\reportNull$ is then executed.

We now describe how each possible contract is executed. Let $\{\Pdist_\theta
\mid \theta \in \Theta \}$ be a family of probability distributions
indexed by $\theta$. These are the possible distributions generating
the data $\evidence$ that the principal uses for decision-making. The family of distributions is known to both the principal and agent.
Without loss of generality, we assume the variable
$\evidence$ is a $p$-value, meaning that distribution under the null
hypothesis is uniform:
\begin{align*}
\Pdist_{\theta_0}(\evidence \leq \threshold) = \threshold \qquad
\mbox{for any choice of threshold $\threshold \in [0,1]$.}
\end{align*}
When
an agent with parameter $\theta$ chooses a statistical contract
$(\fthreshold{\reportNull}, \reward_\reportNull, \cost_\reportNull)$,
he spends $\cost_\reportNull$ dollars to conduct an experiment, which
yields a random variable $\evidence$ drawn from distribution
$\Pdist_\theta$. 
The principal then approves the proposal if
$\evidence \leq \fthreshold{\reportNull}$ and denies it otherwise.

If the
principal approves, the agent receives $\reward_\reportNull$ dollars
in return, resulting in a net payoff of $\reward_\reportNull - \cost_\reportNull$
dollars after accounting for the up-front cost. Otherwise, when the principal denies the proposal, the agent receives no reward, thus losing
a net of $\cost_\reportNull$ dollars. The principal-agent interaction is
summarized as follows:

\begin{tcolorbox}[boxrule=0.5pt,colback=white,colframe=black,sharp corners]
    {\bf Principal-agent hypothesis test with a menu}
\begin{enumerate}[itemsep=0pt]
\item The principal publishes a contract menu $\ContractMenu = \Big \{
  (\fthreshold{\reportNull}, \reward_\reportNull, \cost_\reportNull)
  \mid p \in \reportSupp \Big \}$ indexed by types $p \in \reportSupp
  \subseteq [0,1]$.
    \item Based on his private prior $\priorDist$ and utility, the
      agent reports a type or opts out.
    \item If the agent reports type $\reportNull$, he then receives
      contract $(\fthreshold{\reportNull}, \reward_\reportNull,
      \cost_\reportNull)$ and spends $\cost_\reportNull$ dollars to
      run a trial, which outputs a random variable $\evidence \sim
      \Pdist_\theta$.
    \item The principal approves the proposal if $\evidence \leq
      \fthreshold{\reportNull}$ and denies otherwise.
    \item If the principal approves, the agent receives a reward of
      $\reward_\reportNull$ dollars.
\end{enumerate}
\end{tcolorbox}


\subsection{Utility-maximizing agents}
\label{SecSetup:agent_utility}

We assume that agents are strategic decision-makers who act to maximize expected utility. Each agent is risk-neutral, meaning that their utility is equivalent to their expected gain in wealth. The agent’s utility, determined jointly by their own actions and the principal’s decision, can fall into one of three possible outcomes:
\begin{itemize}[leftmargin=1em]
\item If the agent chooses contract $(\threshold_\reportNull,
\reward_\reportNull, \cost_\reportNull)$ and the principal
approves the proposal, the agent gains $\reward_\reportNull -
\cost_\reportNull$ in wealth.
\item If the agent chooses contract $(\threshold_\reportNull,
\reward_\reportNull, \cost_\reportNull)$ and the principal
denies the proposal, the agent loses $\cost_\reportNull$ in
wealth.
\item If the agent chooses to opt out, he neither gains nor loses
in wealth.
\end{itemize}

For an agent who opts into the contract indexed by $\reportNull$, his
change in wealth is given by the random variable
$\WealthAfter{\reportNull}$. Since the prior distribution $\priorDist$
of any agent is specified by the null probability $\priorNull =
\priorDist(\theta = \theta_0)$, we define the agent’s expected utility
as
\begin{align*}
\utilfunc(\priorNull; \reportNull) \defn \E_{
  \substack{\theta \sim \Qprob \\ \evidence \sim \Pdist_\theta}} \Big[
  \WealthAfter{\reportNull} \Big].
\end{align*}
Given a contract menu $\ContractMenu$ with support $\reportSupp$, we
assume that agents are \emph{wealth-maximizing}, meaning that the
behavior of an agent of type $\priorNull$ is characterized by the
\emph{selection function}
\begin{subequations}
\begin{align}
\label{EqnDefnSelFun}  
\SelFun(\priorNull) & \defn \arg \max_{\reportNull \in \reportSupp}
\utilfunc(\priorNull; \reportNull).
\end{align}
If $\utilfunc(\priorNull; \SelFun(\priorNull)) \geq 0$, the agent
accepts the chosen contract; otherwise, the agent opts out.

To make the utility expression more explicit, recall the type I error
and power functions of a simple-simple hypothesis test with threshold
$\threshold \in [0,1]$:
\begin{align}
\label{eq:null_alt_app}  
\nullAppSimple{\threshold} \defn \Pdist_{\theta_0}(\evidence \leq
\threshold) = \threshold \quad \mbox{and} \quad
\altAppSimple{\threshold} \defn \Pdist_{\theta_1}(\evidence \leq
\threshold).
\end{align}
With this notation, the expected utility of an agent of type
$\priorNull$ selecting contract $\reportNull$ can be written as
\begin{align}
\label{EqnOptInUtility}
\utilfunc(\priorNull;\reportNull) & = \priorNull
\reward_\reportNull \big[\nullAppSimple{\threshold_\reportNull} -
\altAppSimple{\threshold_\reportNull}\big] +
\big[\reward_\reportNull\altAppSimple{\threshold_\reportNull} -
\cost_\reportNull \big].
\end{align}
\end{subequations}
See~\Cref{SecProofOptInUtility} for the proof of this claim.

Notice the utility~\eqref{EqnOptInUtility} is linear in the agent's
prior null probability $\priorNull$ with slope $\reward_\reportNull
[\nullAppSimple{\threshold_\reportNull} -
  \altAppSimple{\threshold_\reportNull}]$. We are interested in the
case with non-trivial power, i.e.,
$\altAppSimple{\threshold_\reportNull} >
\nullAppSimple{\threshold_\reportNull}$, in which case the slope is
negative. Thus, agents with smaller $\priorNull$ (more optimistic
priors) derive greater expected utility from any given
contract. Moreover, the utility is nondecreasing in $\tau_p$, since
both the type I error $\nullAppSimple{\threshold_\reportNull}$ and the
power $\altAppSimple{\threshold_\reportNull}$ are nondecreasing.

\subsection{Statistical decision problem of the principal}
\label{SecSetup:principal_opt}

At a high level, the principal’s objective is to design a menu that
controls the two types of error in a binary hypothesis test.  The menu
should achieve a desired balance between type I errors (where $\theta
= \theta_0$ but the principal approves) and type II errors (where
$\theta = \theta_1$ but the principal denies). We consider two
specifications of this trade-off below: (a) she may wish to minimize a
weighted sum of the errors, or (b) to maximize total discovery rate subject to a constraint on the false discovery rate.  In
either case, the key feature is that the principal's optimal decision
rule depends on the agent’s prior null probability $\priorNull$. When
the null is more likely (large $\priorNull$), a lenient test with a
large threshold risks too many false positives, so the optimal test
should be more stringent. Conversely, when the null is unlikely (small
$\priorNull$), the principal can afford to be more permissive to avoid
unnecessary false negatives.

We use $\tau_q$ to denote the principal's ideal threshold for an agent
with prior $q$, and refer to it as the \emph{type-optimal
threshold}. We describe two concrete cases next.

\paragraph{Weighted sum of type I and type II error.}

In many settings, the principal's cost of making a type I error is not
the same as that of type II error.  Suppose type I errors incur cost
$\FPcost$ and type II errors incur cost $\FNcost$. For an agent of
type $\priorNull$, if the principal deployed a test at level
$\threshold$, the associated Bayes risk (BR) is given by
\begin{align}
\label{EqnBayesRiskOmega}  
\BayesRisk_\omega(\priorNull, \threshold) = \FPcost \cdot \priorNull
\cdot \underbrace{\Pdist_{\theta_0}(\evidence \le \threshold)}_{\equiv
  \threshold} + \FNcost \cdot (1-\priorNull) \cdot
\underbrace{\Pdist_{\theta_1}(\evidence > \threshold)}_{\equiv 1-
  \beta_1(\threshold)}.
\end{align}
Classical results show that the optimal decision rule for the Bayes
risk $\BayesRisk_\omega$ is to threshold the likelihood ratio
statistic at level $\priorNull \FPcost / (1 - \priorNull)
\FNcost$. Thus, the optimal $p$-values $\evidence$ for the principal
are those obtained from the likelihood ratio statistic.
Concretely, if we let $\evidence \mapsto \mathcal{L}(\evidence)$ be
the function mapping $\evidence$ to the likelihood ratio, then the
type-optimal threshold on the $p$-value scale is
\begin{align}
\label{EqnBayesThreshold}
\fthreshold{\priorNull} \equiv \fthreshold{\priorNull}(\omega) \defn
\mathcal{L}^{-1}\bigg( \frac{\priorNull \FPcost}{(1-\priorNull)
  \FNcost}\bigg).
\end{align}

\paragraph{Maximizing TDR subject to FDR control.} Alternatively,
the principal might wish to maximize the expected true discovery rate
(TDR) while controlling the false discovery rate (FDR) at a target
level $\fdrLevel$.  For an agent of type $\priorNull$, the FDR
associated with the hypothesis test with threshold $\threshold$ is
given by
\begin{subequations}
\label{EqnFDRModel}  
  \begin{align}
\label{EqnDefnFDR}    
    \fdr(\priorNull, \threshold) & \defn \Prob \Big[\theta \in
      \Theta_0 \mid \evidence \le \threshold \Big] = \frac{\priorNull
      \nullAppSimple{\threshold}}{\priorNull
      \nullAppSimple{\threshold} + (1-\priorNull)
      \altAppSimple{\threshold}},
\end{align}
which corresponds to the probability that a proposal approved is
actually ineffective.  On the other hand, the TDR for a given
$(\priorNull, \threshold)$-pair is given by
\begin{align}
\label{EqnDefnETD}
\etd(\priorNull, \threshold) \defn (1 - \priorNull)
\altAppSimple{\threshold},
\end{align}
\end{subequations}
which represents the probability of correctly approving non-null
proposals from agent $\priorNull$ when using threshold
$\threshold$. Since the FDR is increasing in $\threshold$, the
constraint $\fdr(\priorNull, \threshold) \leq \fdrLevel$ places an
upper bound on the threshold $\threshold$ that agent type $\priorNull$
can be allowed to choose. The TDR is also increasing in $\threshold$,
and thus maximized by choosing the largest such $\threshold$.  We
conclude that the type-optimal threshold for an agent of type
$\priorNull$ is given by
\begin{align}
\label{EqnFDRThreshold}
\fthreshold{\priorNull} \equiv \threshold_\priorNull(\fdrLevel) \defn
\sup \bigg\{\threshold \, \mid \, \fdr(\priorNull, \threshold) \leq
\fdrLevel \bigg\}.
\end{align}

\paragraph{Defining the oracle performance:}
In either case, we are left with the following problem: if the
agent types were known, then given an agent of type $\priorNull$, the principal would implement the test with type-optimal threshold
$\fthreshold{\priorNull}$.  Doing so over the full population of
agents (as $\priorNull$ varies) would achieve a statistical guarantee
that we refer to as the \emph{heterogeneous agent oracle performance}.  More specifically,
given any distribution $\TypeDist$ over the set of possible agent
types $\priorNull \in [0,1]$, the oracle value of the
$\omega$-weighted Bayes risk is given by
\begin{subequations}
  \begin{align}
\label{EqnBayesRiskOracle}    
\int_0^1 \BayesRisk_\omega \big(\priorNull,
\threshold_\priorNull(\omega) \big) d \TypeDist(\priorNull),
\end{align}
based on the type-optimal threshold $\threshold_\priorNull(\omega)$
from equation~\eqref{EqnBayesThreshold}.  In words, this is the
minimum of the $\omega$-Bayes-risk over all possible tests achievable
when agent types are known.

Similarly, for FDR controlled at level $\fdrLevel$, the oracle value
of the TDR is given by
\begin{align}
\label{EqnFDROracle}      
\int_0^1 \etd \big(\priorNull, \threshold_\priorNull(\fdrLevel) \big)
d \TypeDist(\priorNull),
\end{align}
\end{subequations}
based on the type-optimal threshold $\threshold_\priorNull(\fdrLevel)$
from equation~\eqref{EqnFDRThreshold}.  In words, this is the maximum
of the $\tdr$ subject to the $\fdr$ being $\fdrLevel$-bounded over all
possible tests when agent types are known.

Of course, the key challenge is that types are unobserved. Agents,
motivated by their own payoffs, may misreport their type to secure
more favorable terms. The principal’s problem is therefore to design a
contract menu that simultaneously uses the type-optimal thresholds
while also making truthful reporting optimal for the agent. We refer
to this object as a separating menu, since it acts to separate agents
according to their type. If a separating menu can be constructed, it
allows the principal to match the oracle statistical performance, as
defined in equations~\eqref{EqnBayesRiskOracle}
and~\eqref{EqnFDROracle}, \emph{without knowing the types or their proportion $\TypeDist$ a priori}.
Designing such separating menus is the focus of the next section.


\section{Main Results}
\label{SecMain}

In this section, we develop the main results: how the principal can
construct menus $\ContractMenu$ that implement type-optimal thresholds
$\fthreshold{\priorNull}$. In~\Cref{SecMain:psr_menu}, we begin with a
formal definition of a separating menu, identifying the
incentive compatibility and participation constraints for truthful
reporting. We then show how to construct such menus for arbitrary
agent distributions, establish their link to proper scoring rules, and
examine implementation costs. Finally, \Cref{SecMain:opt_menu}
considers optimal designs under two practical criteria: minimizing
financial cost and handling constrained contract parameters.

Throughout, we assume the principal knows the statistical properties
of the test, namely the type I error rate $\nullAppSimple{\threshold}
= \threshold$ and the power function $\threshold \mapsto
\altAppSimple{\threshold}$ (cf. equation~\eqref{eq:null_alt_app}). We
restrict attention to cases where the power is non-trivial:
\begin{equation}
\label{eq:nontrivial_power}
\altAppSimple{\threshold} > \nullAppSimple{\threshold} \qquad
\mbox{for all $\threshold \in (0,1)$.}
\end{equation}
We assume that the principal's type-optimal threshold assignment $q
\mapsto \tau_q$ is non-increasing but make no further restrictions. As
such, the menu construction that follows applies when the principal
wishes to control the weighted combination of type I and type II errors
or to maximize power given an FDR constraint.


\subsection{Separating menus and oracle risk}
\label{SecMain:psr_menu}

As discussed in~\Cref{SecSetup:principal_opt}, if the principal could
observe each agent's true type, she would simply assign the
type-optimal threshold $\fthreshold{\priorNull}$. However, since types
are private, type-optimal thresholds cannot be assigned directly.
Instead, the principal must design a menu of contracts---each
specifying a $p$-value threshold, a reward, and a cost---that
incentivizes agents to truthfully reveal their types.

We now formalize the requirements for such menus. Given a subset
$\mathcal P \subset [0,1]$ of agent types, a contract menu
$\ContractMenu$ is said to be \emph{separating} if each agent opts
into the contract designed for their true type.  In terms of the
selection function~\eqref{EqnDefnSelFun}, a menu $\ContractMenu$ is
separating if and only if
\begin{align}
\label{def:feasible_menu}
 \SelFun(\priorNull) \stackrel{(a)}{=} \priorNull \quad \mbox{and}
 \quad \utilfunc(\priorNull; \SelFun(\priorNull)) \stackrel{(b)}{\geq}
 0 \qquad \mbox{for all $\priorNull \in \mathcal P$.}
\end{align}

We refer to condition \eqref{def:feasible_menu}(a) as an
\emph{incentive compatibility} constraint, which implies that it is
optimal for any agent type to report their type $\priorNull$
truthfully, since misreporting will not yield higher utility.
Condition \eqref{def:feasible_menu}(b) is a \emph{participation}
constraint: since agents can always opt out, their selected contract
should guarantee nonnegative expected utility.

Thus, if the principal can construct a separating menu $\ContractMenu$
with type-optimal threshold $\fthreshold{\priorNull}$, then for each
type of agent that opts in, she can conduct statistically optimal
tests.  We now turn to the construction of separating menus. In what
follows, we show that every separating menu is characterized by a
real-valued convex function $\Gfun$ and establish a connection between
separating menus and proper scoring rules.

\subsubsection{Main result: Characterization of separating menus}
\label{SecMain:characterize}

We first set up the notation needed to provide our general
characterization. Consider a function $\Gfun: \supp(\typeDist) \to \R$
defined on the support of the type distribution $\typeDist$ that
satisfies the following two properties.  First, for each \(
\reportNull \in \supp(\typeDist) \), there exists a scalar
$\gstar{\reportNull} < 0$ such that
      \begin{subequations}
        \label{EqnGfun}
    \begin{align}
    \label{eq:Gfun_ic}
        \Gfun(\priorNull) > \Gfun(\reportNull) +
        \gstar{\reportNull}(\priorNull - \reportNull) \quad \text{for
          all } \priorNull \in \supp(\typeDist) \setminus
        \{\reportNull\}.
    \end{align}
    Second, we have
      \begin{align}
        \label{eq:Gfun_participation}        
        \Gfun(\worstPriorNull) \geq 0, \quad \text{where }
        \worstPriorNull \defn \sup \{q \mid q \in \supp(\typeDist) \}.
      \end{align}
      \end{subequations}
When $\supp(\typeDist)$ is an interval in $[0,1]$ and $\Gfun$ is
differentiable, these conditions are equivalent to $\Gfun$ being
strictly convex, nonnegative, and decreasing. In the
non-differentiable case, we can understand $\gstar{\reportNull}$ as
playing the role of a subgradient of $\Gfun$ at $\reportNull$, and
when $\supp(\typeDist)$ is a discrete set, a function $\Gfun$
satisfying~\eqref{eq:Gfun_ic} can be interpreted as satisfying a form
of discrete convexity.  \\

\noindent Our main result is that the class of all such functions
$\Gfun$ characterizes the set of separating menus:
\begin{theorem}
\label{th:psr_menu}
Consider a function $\Gfun: \supp(\typeDist) \to \R$ that satisfies
conditions~\eqref{eq:Gfun_ic} and~\eqref{eq:Gfun_participation}, and a
collection of type-optimal thresholds $\{ \threshold_\priorNull \;
\mid \; \priorNull \in \supp(\typeDist) \}$.
\begin{enumerate}[label=(\alph*), ref=\thetheorem(\alph*)]
    \item By defining the reward-cost pairs
\begin{align}
\label{eq:psr_reward_cost}
\reward_\reportNull & =
\frac{\gstar{\reportNull}}{\nullAppSimple{\threshold_\reportNull} -
  \altAppSimple{\threshold_\reportNull}} \quad \mbox{and} \quad
\cost_\reportNull = \gstar{\reportNull} \Big[
  \frac{\altAppSimple{\threshold_\reportNull}}{\nullAppSimple{\threshold_\reportNull}
    - \altAppSimple{\threshold_\reportNull}} + \reportNull \Big] -
\Gfun(\reportNull),
\end{align}
we obtain a contract menu $\ContractMenu \defn \{
(\threshold_\reportNull, \reward_\reportNull, \cost_\reportNull) \mid
\reportNull \in \supp(\typeDist) \}$ that is separating.
\item Conversely, any separating menu indexed by $\reportNull$ ranging
  over $\supp(\typeDist)$, can be constructed in this way for some
  function $\Gfun: \supp(\typeDist) \to \R$ satisfying
  conditions~\eqref{eq:Gfun_ic} and~\eqref{eq:Gfun_participation}.
\end{enumerate}
\end{theorem}
\noindent See~\Cref{proof:psr_menu} for the proof of this result.  We
also provide an illustration of this construction when the agent types
are finite in~\Cref{SecApp:finite_type_menu}. \\

\Cref{th:psr_menu} provides a systematic way to construct separating
menus for any distribution over agent types. Moreover, it shows that
there is a one-to-one correspondence between convex functions $\Gfun$
and separating menus.

Given this one-to-one correspondence, one might suspect that $\Gfun$
has a fundamental meaning.  As we show in~\Cref{proof:psr_menu}, the
function value $\Gfun(\priorNull)$ represents the expected utility
attained by an agent of type $\priorNull$ when reporting
truthfully---that is, we have the equivalence $\Gfun(\priorNull) =
\utilfunc(\priorNull;\priorNull)$.  The reward
function~\eqref{eq:psr_reward_cost} is constructed so that when an
agent of type $\priorNull$ reports type $\reportNull$, his
utility~\eqref{EqnOptInUtility} induced by the contract
$(\threshold_\reportNull, \reward_\reportNull, \cost_\reportNull)$ has
a slope given by \mbox{$\gstar{\reportNull} =
  \reward_\reportNull[\nullAppSimple{\threshold_\reportNull} -
    \altAppSimple{\threshold_\reportNull}]$.} In parallel, the cost
function is chosen so that his utility~\eqref{EqnOptInUtility} under
this contract coincides with the supporting hyperplane of $\Gfun$ at
$\reportNull$, given by
\begin{align*}
  \utilfunc(\priorNull; \reportNull) = \Gfun(\reportNull) +
  \gstar{\reportNull}(\priorNull - \reportNull) \quad \mbox{for
    $\priorNull \in [0,1]$.}
\end{align*}
As a result, property~\eqref{eq:Gfun_ic} guarantees incentive
compatibility~\eqref{def:feasible_menu}(a): each agent type achieves
maximal expected utility by reporting truthfully.

The condition $\gstar{\priorNull} < 0$ is
imposed to maintain consistency with the non-trivial power
assumption~\eqref{eq:nontrivial_power}; it implies that utility
declines with type $\priorNull$, so more optimistic agents (those with
lower $\priorNull$) obtain higher expected utility under the same
contract. Lastly, property~\eqref{eq:Gfun_participation} ensures that
participation constraints~\eqref{def:feasible_menu}(b) are met: once the highest type secures
nonnegative utility, all agents with lower types will strictly prefer to
participate.

\subsubsection{Matching the oracle statistical performance}
\label{SecMain:oracle}
The statistical motivation of our work was to determine when it is
possible for the principal to match the statistical performance of the
oracle given access to agent types a priori.  In particular, recall
the definitions of the oracle risk for the $\omega$-weighted Bayes
risk~\eqref{EqnBayesRiskOracle}, and the TDR
risk~\eqref{EqnFDROracle}.  A straightforward but important
consequence of~\Cref{th:psr_menu} is to provide a prescriptive means
for the principal to match the oracle.  In stating this result, we say
that a function $\Gfun$ is valid if it satisfies the conditions
of~\Cref{th:psr_menu}.  We summarize as follows:
\begin{corollary}[Matching the oracle performance]
\label{CorOracle}  
Consider any valid $\Gfun$ and any agent population.
  \begin{enumerate}[label=(\alph*), ref=\thetheorem(\alph*)]
  \item[(a)] Minimizing the $\omega$-weighted Bayes
    risk~\eqref{EqnBayesRiskOmega}: The menu constructed
    from~\Cref{th:psr_menu}(a) using the type-optimal
    thresholds~\eqref{EqnBayesThreshold} yields a testing protocol
    that matches the oracle performance~\eqref{EqnBayesRiskOracle}.
  \item[(b)] Maximizing $\tdr$ subject to $\fdr$
    constraint~\eqref{EqnFDRModel}: The menu constructed
    from~\Cref{th:psr_menu}(a) using the type-optimal
    thresholds~\eqref{EqnFDRThreshold} yields a testing protocol that
    matches the oracle performance~\eqref{EqnFDROracle}.    
  \end{enumerate}
\end{corollary}

\begin{proof}
  The proof of each claim follows a parallel argument: here we
  prove~\Cref{CorOracle}(a).  First, the construction
  of~\Cref{th:psr_menu}(a) ensures that we have a separating menu.  By
  property~\eqref{def:feasible_menu}(a) from the definition of a
  separating menu, it follows that each agent truthfully reports its
  own type $\priorNull$.  Moreover, since the menu construction was
  based on the type-optimal thresholds~\eqref{EqnBayesThreshold}, the
  principal's protocol ensures that an agent of type $\priorNull$ is
  assigned the test with the matched type-optimal threshold.
  Therefore, the overall testing procedure has a Bayes risk, when
  averaged over any agent population
  $\TypeDist$, that matches the oracle
  performance~\eqref{EqnBayesRiskOracle}.
\end{proof}

Thus, we have shown that it is always possible for the principal to
match the oracle performance.  Note that there are no assumptions on
the structure of the underlying testing problem, apart from having a
$p$-value uniform under the null, as well as non-trivial
power~\eqref{eq:nontrivial_power}.  The caveat here is that~\Cref{th:psr_menu},
and hence~\Cref{CorOracle}, assumes that the principal has full control over
all three components (cost, reward, and threshold) of each statistical contract.
In~\Cref{SecMain:opt_menu_constR}, we visit a constrained setting in which
the principal has only partial control, and see that the answer then depends
on more fine-grained features of the test.

\subsubsection{Connection to proper scoring rules} 
\label{SecMain:connection_to_psr}

A separating menu ensures that any utility-maximizing agent type opts
into the contract designed for their true type. In this way, the
principal can elicit an agent’s private belief by observing their
choice. This resembles the logic of proper scoring rules: reward
functions that yield truthful elicitation under expected utility
maximization. In what follows, we demonstrate that an
incentive-compatible menu can be interpreted as inducing a strictly
proper scoring rule.

In order to make this connection precise, consider the binary random
variable $Y \defn \indicator(\theta = \theta_0)$, where $Y=1$
indicates that the agent's proposal is ineffective ($\theta =
\theta_0$) and $Y=0$ indicates effectiveness. The agent’s belief is
summarized by the prior null probability $\priorNull =
\priorDist(\theta = \theta_0)$, while the reported type is denoted by
$\reportNull$. For each report $\reportNull$, the principal offers a
contract consisting of a reward $\reward_\reportNull$, a $p$-value
threshold $\threshold_\reportNull$, and a cost
$\cost_\reportNull$. The agent’s expected payoff depends on the chosen
contract and the state $Y$ as
\begin{align}
\label{EqnFullScore}
\Score(\reportNull,Y) \defn 
\begin{cases}
\reward_\reportNull \nullAppSimple{\threshold_\reportNull} -
\cost_\reportNull \quad \mbox{when} \quad Y=1 \\ \reward_\reportNull
\altAppSimple{\threshold_\reportNull} - \cost_\reportNull \quad
\mbox{when} \quad Y=0
\end{cases}.
\end{align}
This function can be interpreted as a scoring rule: the agent reports
a probability $\reportNull$ and receives a payoff that depends on the
realized outcome $Y$. If the agent’s true belief is $\priorNull$, his
expected payoff from reporting $\reportNull$ is
$\Score(\reportNull, \priorNull) = \priorNull \big[\reward_\reportNull
  \nullAppSimple{\threshold_\reportNull} - \cost_\reportNull \big] +
(1-\priorNull) \big[\reward_\reportNull
  \altAppSimple{\threshold_\reportNull} - \cost_\reportNull \big]$,
which coincides
with our previously defined utility
function~\eqref{EqnOptInUtility}.

A scoring rule is said to be \emph{proper} if truthful reporting maximizes expected payoff, meaning that
\begin{align*}
\Score(\reportNull, \priorNull) \leq \Score(\priorNull, \priorNull)
\quad \mbox{for all pairs $\reportNull, \priorNull \in [0,1]$,}
\end{align*}
and \emph{strictly proper} if equality holds if and only if
$\reportNull = \priorNull$. By construction, an incentive-compatible
menu guarantees that each type $\priorNull$ strictly prefers the
contract intended for them. Equivalently, it induces a scoring rule
that is strictly proper. Thus, the design of an incentive-compatible
menu in the hypothesis testing context can be viewed as the design of
a strictly proper scoring rule tailored to the principal's statistical
performance criterion. This connection also clarifies why separating
menus can be constructed using convex real functions: every proper
scoring rule admits a convex potential
representation~\cite{mccarthy1956measures, savage1971elicitation,
  gneiting2007strictly}.

Our setting, however, differs from the abstract scoring-rule framework
in the following way: the principal has only \emph{partial control} of
the agent's utility via the menu design.  As evident from
definition~\eqref{EqnFullScore}, the induced scoring rule depends not
only on contract parameters $(\threshold_\reportNull,
\reward_\reportNull, \cost_\reportNull)$ but also on the type I error
and power function determined by the hypothesis test itself. The
principal can adjust contract parameters so that the resulting utility
behaves like a proper scoring rule, but she cannot fully dictate
payoffs independently of the test. This distinction foreshadows an
important limitation: if additional constraints are imposed on the
contract parameters, it may no longer be possible to induce a strictly
proper scoring rule. We illustrate this in
\Cref{SecMain:opt_menu_constR}, where requiring the reward to be
constant restricts the class of hypothesis tests and the range of
agent types for which the principal can implement separating menus.

\subsubsection{The costs of elicitation}
\label{SecMain:cost_of_elicitation}

Constructing separating menus enables the principal to assign type-optimal thresholds and elicit agents’ private beliefs. However, this design is not costless: in order to satisfy incentive compatibility, the principal must leave some surplus to agents. To understand what is lost by eliciting information through menus, we introduce two complementary measures of the trade-off, each defined with respect to a different benchmark.

\paragraph{Information rent.} The first benchmark is the full-information setting, in which the principal directly observes each agent’s type. With full information, she could implement the first-best contract by assigning the type-optimal threshold and adjusting the reward and cost so that each agent earns zero utility. Agents would still participate, but they would not retain any surplus. 

However, under asymmetric information, the principal cannot
condition contracts directly on types and must instead design a menu that induces truthful reporting. To
do so, she must offer higher utility to agents with lower prior null to prevent misreporting. This additional utility is the information rent---the surplus that agents receive in order to implement incentive compatibility (see~\cite{laffont2001,bolton2004contract}). 
In our setting, the utility that an agent of type $\priorNull$ gains
from participating is exactly $\Gfun(\priorNull)$, so
the \emph{total information rent} for a separating menu constructed by $\Gfun$ is 
\begin{align}
\label{EqnDefnInfoRent}  
 \infoRent & \defn \int_0^1 \Gfun(\priorNull) d \typeDist(\priorNull).
\end{align}
When the type distribution $\typeDist$ is uniform on $[0,1]$, the
information rent corresponds to the area under the curve $\Gfun$.

\paragraph{Screening cost.} Information rent highlights the gap to the unattainable first-best.
Since it is the unavoidable price of asymmetric information, a more
practical benchmark is the pooling contract, where the principal
offers the same contract to all agents. The natural candidate is a
baseline contract designed for the worst agent type $\worstPriorNull$
(the highest prior null), since this ensures participation of all
agent types.  To ensure incentive compatibility, however, the menu
must give each agent with better type (that is, with lower
$\priorNull$) strictly higher utility than he would obtain from
pretending to be type $\worstPriorNull$. In this situation, the
principal must sacrifice some surplus compared to simply offering a
single contract. Following the contract theory literature, we refer to
this surplus loss as the screening cost, since it is the cost of
designing a menu that separates (or ``screens") types rather than
pooling them. Formally, we define the \emph{screening cost} as
\begin{align}
\label{EqnDefnSeparation} 
\costSep & \defn \int_0^1 \big \{ \Gfun(\priorNull) -
\utilfunc(\priorNull; \worstPriorNull) \big \} d
\typeDist(\priorNull),
\end{align}
where $\utilfunc(\priorNull; \worstPriorNull)$ is the utility an agent
of type $\priorNull$ would receive from the contract intended for
$\worstPriorNull$. \\

In classical contract theory, there is typically a trade-off: a menu enables more precise targeting of types, but at the expense of screening costs; a single contract avoids such costs but results in inefficiencies from misallocation. One might expect the same tension here. Surprisingly, as we show later in \Cref{co:Gfun_varyingR}, there exist separating menus that achieve incentive compatibility with arbitrarily small financial cost to the principal. In some situations, the principal even realizes financial gains by offering separating menus rather than a single contract.

The key difference in our setting lies in how utility is transferred across types. In standard contracting problems, transfers are monetary rewards or cost adjustments, which reduce the principal’s payoff. In our setting, however, much of the utility increase for better-type agents comes from being assigned looser statistical thresholds, which raise the probability that their proposals are accepted. 

This is not zero-sum; assigning more permissive thresholds to better types increases the likelihood of true discoveries while maintaining control of type I errors, which simultaneously benefits both the principal and the agent. Menu-based testing can therefore be Pareto-improving: agents receive strictly higher utility when presented with more options, while the principal is able to achieve better statistical performance.

\subsection{Optimal menu design}
\label{SecMain:opt_menu}

\Cref{th:psr_menu} characterizes the full class of separating menus
that can elicit agents’ private beliefs and implement type-optimal
thresholds. Yet this characterization leaves open a practical
question: which separating menu should the principal actually
implement? Because infinitely many functions $\Gfun$ induce menus that
satisfy participation and incentive compatibility, additional design
criteria are needed to guide selection.

We focus on two such criteria, each linked to the discussion in
earlier sections. First, as noted in
\Cref{SecMain:cost_of_elicitation}, compared to offering a single
contract, a separating menu generally entails a screening cost to the
principal. A natural question is therefore whether menus can be
constructed to minimize this cost. In
\Cref{SecMain:opt_menu_varyingR}, we show that there exist families of
$\Gfun$ that make the screening cost arbitrarily small, effectively
rendering elicitation ``free." Second, as emphasized in
\Cref{SecMain:connection_to_psr}, separating menus must induce
strictly proper scoring rules. This requires flexibility in contract
parameters, and when some parameters are fixed exogenously, the set
of separating menus can shrink dramatically. In
\Cref{SecMain:opt_menu_constR}, we examine the case where rewards are
held constant and show that under this restriction, there is a unique
separating menu.

Throughout this section, we restrict attention to differentiable
functions $\Gfun$ defined on $[0,\worstPriorNull]$, where
$\worstPriorNull$ denotes the least favorable agent type that the principal is willing to accommodate.  Any $\Gfun$ satisfying
conditions~\eqref{eq:Gfun_ic} and~\eqref{eq:Gfun_participation} is
then differentiable, strictly convex, and nonnegative on this
interval.  Since separating menus guarantee statistical efficiency by
construction, our analysis focuses on their financial cost relative to
a benchmark. As in the definition~\eqref{EqnDefnSeparation} of the
screening cost, we compare each separating menu to a base contract
$(\threshold_{\worstPriorNull}, \reward_{\worstPriorNull},
\cost_{\worstPriorNull})$ designed for the worst type
$\worstPriorNull$. To ensure participation of all types up to
$\worstPriorNull$ and non-participation of types exceeding
$\worstPriorNull$, we impose a \emph{worst-case participation
condition} on the base contract
\begin{align}
\label{eq:worst_participation}
\utilfunc(\worstPriorNull; \threshold_{\worstPriorNull},
\reward_{\worstPriorNull}, \cost_{\worstPriorNull}) = 0,
\end{align}
so that $\Gfun(\worstPriorNull) = 0$.

\subsubsection{Separating menus with minimal screening cost}
\label{SecMain:opt_menu_varyingR}

Implementing a separating menu to elicit agents’ private beliefs
generally entails a screening cost to the principal. By
definition~\eqref{EqnDefnSeparation}, this cost is the extra utility
agents receive relative to the base contract.  From the principal’s
perspective, screening cost is the expected financial loss that arises
when operationalizing a menu rather than offering a single base
contract.  To see this, recall the benchmark in which the principal
offers a fixed base contract $(\threshold_{\worstPriorNull},
\reward_{\worstPriorNull}, \cost_{\worstPriorNull})$ to all
agents. Under this contract, the principal pays a reward
$\reward_{\worstPriorNull}$ upon approval and collects a fixed payment
$\cost_{\worstPriorNull}$ from each agent to cover trial costs. By
contrast, under a separating menu, the principal tailors the contract
terms to each type. For an agent of type $\priorNull$, the principal
may adjust both the statistical threshold and the financial terms:
additional rewards $\reward_\priorNull - \reward_{\worstPriorNull}$
upon approval, and extra payments $\cost_\priorNull -
\cost_{\worstPriorNull}$ from the agent. These adjustments
differentiate the menu from the base contract and implement
type-specific targeting.

The principal’s expected financial return from offering the tailored
contract $(\threshold_\priorNull,\reward_\priorNull,\cost_\priorNull)$
to type $\priorNull$, relative to the base contract, is
\begin{align}
\label{DefPrincipalReturn}
& \underbrace{\big[\cost_\priorNull - \reward_\priorNull \cdot
    \Prob(\mbox{approve agent } \priorNull \mbox{ under }
    \threshold_\priorNull)\big]}_{\mbox{type-tailored contract}} -
\underbrace{\big[\cost_{\worstPriorNull} - \reward_{\worstPriorNull}
    \cdot \Prob(\mbox{approve agent } \priorNull \mbox{ under }
    \threshold_{\worstPriorNull})\big]}_{\mbox{base contract}},
\end{align}
which simplifies to $- \big[\Gfun(\priorNull) -
  \utilfunc(\priorNull;\worstPriorNull)\big]$. Integrating over the
agent population $\typeDist$, the principal’s expected return is therefore the negative of the screening
cost~\eqref{EqnDefnSeparation}. Since incentive compatibility requires
the screening cost to be strictly positive, the principal necessarily
incurs a financial loss from implementing a separating menu. In
practice, resource constraints may make such losses problematic. Thus,
an important design goal is to construct menus that minimize screening
cost, which is equivalent to minimizing the principal’s financial loss
while still guaranteeing incentive compatibility.

Using~\Cref{th:psr_menu}, we can deduce an entire family of
cost-reward pairs that yield separating menus. By careful choice,
these menus can be designed to incur an arbitrarily small screening
cost.  In particular, let $z \mapsto \specfunc(z)$ be any
differentiable function such that
\begin{subequations}
\begin{align}
\label{EqnSpecialAssumptions}
\specfunc(\worstPriorNull) = 0 \quad \mbox{and} \quad \specfunc(z) > 0 \quad \mbox{and} \quad \specfunc'(z) < 0 \qquad
\mbox{for all $z \in [0, \worstPriorNull)$,}
\end{align}
and define the reward-cost pairs
\begin{align}
\label{eq:varyingR_reward}  
\reward_\reportNull & = \bigg( \frac{\reward_{\worstPriorNull}\big[
    \altAppSimple{\threshold_{\worstPriorNull}} -
    \nullAppSimple{\threshold_{\worstPriorNull}}\big]}{\altAppSimple{\threshold_\reportNull}
  - \nullAppSimple{\threshold_\reportNull}} \bigg) \big[ 1 +
  \specfunc(\reportNull) \big], \\
\label{eq:varyingR_cost}    
\cost_\reportNull & = \reward_\reportNull \big[\reportNull
  \nullAppSimple{\threshold_\reportNull} + (1 - \reportNull)
  \altAppSimple{ \threshold_\reportNull} \big] -
\reward_{\worstPriorNull} \big[ \altAppSimple{
    \threshold_{\worstPriorNull}} - \nullAppSimple{
    \threshold_{\worstPriorNull}} \big]
\int_{\reportNull}^{\worstPriorNull} \big[1 + \specfunc(z) \big] dz.
\end{align}
\end{subequations}

\begin{corollary}
  \label{co:Gfun_varyingR}
Consider a base contract $(\threshold_{\worstPriorNull},
\reward_{\worstPriorNull}, \cost_{\worstPriorNull})$ that satisfies
the worst-case participation
constraint~\eqref{eq:worst_participation}, a contract menu
$\ContractMenu$ with type-optimal thresholds $\threshold_\priorNull$
for $\priorNull \in [0, \worstPriorNull]$, and the reward-cost pairs
given by~\eqref{eq:varyingR_reward}--\eqref{eq:varyingR_cost}.  Then
the resulting menu is separating, and it has screening
cost~\eqref{EqnDefnSeparation} relative to the base contract can be
made arbitrarily small by appropriate choice of $\specfunc$.
\end{corollary}
\noindent See \Cref{proof:Gfun_varyingR} for the proof. \\

To elaborate upon the connection to~\Cref{th:psr_menu}, designing a
separating menu with a small screening cost is equivalent to designing
a function $\Gfun$---one which satisfies the properties of the
theorem---such that it lies just above the utility curve
$\PsiBase(\priorNull) \defn \utilfunc(\priorNull;\worstPriorNull)$
under the base contract.  So that the theorem can be applied, the
function $\Gfun$ should be strictly convex. There are many possible
choices of $\Gfun$ that satisfy this requirement.  Given a function
$\specfunc$ satisfying the conditions~\eqref{EqnSpecialAssumptions},
one choice of $\Gfun$---the one used in our proof---is
\begin{align}
  \label{eq:varyingR_G}
  \Gfun(\priorNull) = \reward_{\worstPriorNull}\big[
    \altAppSimple{\threshold_{\worstPriorNull}} -
    \nullAppSimple{\threshold_{\worstPriorNull}}\big]
  \int_{\priorNull}^{\worstPriorNull} \, \big[1 + \specfunc(z) \big]
  dz.
\end{align}
By choosing $\epsilon$ to be small, the principal achieves statistical
efficiency while incurring negligible financial loss relative to the
single-contract benchmark. Thus, the principal effectively elicits
agents' private types for free.

\subsubsection{Separating menus with fixed rewards}
\label{SecMain:opt_menu_constR}

So far, our construction of separating menus has assumed that the
principal retains full control over all components of
contracts---namely, the $p$-value threshold, reward, and
cost. However, when one or more parameters are constrained, it may no
longer be possible to induce a strictly proper scoring rule for all
agent types. Such constraints are not merely theoretical: in practice,
rewards or costs may be determined by forces outside the principal’s
control. For example, in regulatory settings, the reward may be
determined by market forces rather than by the regulator. In our
running example of drug testing, the market generates the revenue for
an approved drug, which is typically similar across firms producing
treatments for the same condition. While the Food and
Drug Administration (FDA) cannot set these
rewards, it could, in principle, adjust the cost by imposing
additional trial requirements or requiring a financial down payment.

With this motivation, we now study what separating menus can be constructed when contract parameters are constrained. Since the $p$-value thresholds are dictated by the principal’s statistical objective, the natural candidates are rewards or costs. In this section, we focus on the case in which all contracts share the same reward. The case of fixed costs is less compelling: it requires the power function $\altAppSimple{\cdot}$ to satisfy a form of convexity (see discussion in \Cref{SecAppFixedCostMenu}), which is unnatural since power functions are usually concave in $\threshold$.

\paragraph{Assumptions.} Fixing rewards makes the design of separating menus more challenging. In particular, we require additional structure on the power function $\altAppSimple{\threshold}$. Our result applies when $\altAppSimple{\threshold}$ satisfies two conditions:
\begin{subequations}
  \label{EqnPowerCond}
    \begin{align}
      \label{EqnPowerCondA}
   \mbox{The function $\altAppSimple{\cdot}$ is concave and
     differentiable, and} \\
\label{EqnPowerCondB}
\beta_1^\prime(\threshold_\priorNull) > 1 \quad \mbox{for all} \quad
\priorNull \in [\bestPriorNull, \worstPriorNull].
    \end{align}
Condition~\eqref{EqnPowerCondA} is not particularly restrictive: by allowing randomized tests, the set of achievable $(\threshold,\altAppSimple{\threshold})$ pairs form a convex subset of $[0,1]\times[0,1]$, whose upper boundary (the power function) is concave and attained by likelihood ratio tests. Condition~\eqref{EqnPowerCondB} is more stringent, and we return to its interpretation below.

Lastly, we also assume that the type-optimal threshold mapping $\priorNull \mapsto \fthreshold{\priorNull}$ is strictly decreasing:
\begin{align}
\label{EqnStrictMonotone}
    \threshold^\prime_\priorNull < 0.
\end{align}
\end{subequations}
This is a mild requirement. For Bayes-optimal decision rules, the mapping~\eqref{EqnBayesThreshold} is strictly decreasing whenever the likelihood ratio is continuous and strictly monotone. For the FDR-control rule~\eqref{EqnFDRThreshold}, strict monotonicity follows under condition~\eqref{EqnPowerCondA} (see proof in \Cref{SecProofLemMonotone}).

\begin{corollary}
  \label{co:Gfun_constR} 
Consider a contract menu $\ContractMenu$ with a constant reward $\reward > 0$, type-optimal thresholds $\threshold_\priorNull
$ for $\priorNull \in [\bestPriorNull, \worstPriorNull]$, and the costs
\begin{align}
\label{eq:constR_cost}    
  \cost_\reportNull = \reward \big[\reportNull
    \nullAppSimple{\threshold_\reportNull} + (1-\reportNull)
    \altAppSimple{\threshold_\reportNull} \big] - \reward
  \int^{\worstPriorNull}_\reportNull \altAppSimple{\threshold_z} -
  \nullAppSimple{\threshold_z} dz.
\end{align}
Then under conditions~\eqref{EqnPowerCondA}--\eqref{EqnPowerCondB} on
the power function and condition~\eqref{EqnStrictMonotone} on the
threshold mapping, this defines the unique separating menu with
constant reward $\reward$.
\end{corollary}
\noindent See~\Cref{proof:Gfun_constR} for the proof. \\

In this case, the connection to~\Cref{th:psr_menu} is provided by the
function
\begin{align}
\label{eq:constR_G}
\Gfun(\priorNull) & \defn \reward \int_{\priorNull}^{\worstPriorNull}
\big[\altAppSimple{\threshold_z} - \nullAppSimple{\threshold_z} \big]
dz, \quad \mbox{with derivative} \quad
\Gfun'(\priorNull) = \reward \big[
  \nullAppSimple{\threshold_\priorNull} -
  \altAppSimple{\threshold_\priorNull} \big].
\end{align}
Since we assume $\priorNull \mapsto \threshold_{\priorNull}$ is
strictly decreasing, condition~\eqref{EqnPowerCondB} ensures
$\Gfun''(\priorNull) \defn \reward\big[1 -
  \beta_1'(\threshold_{\priorNull})\big] \threshold'_{\priorNull}$ is
strictly positive, so $\Gfun$ is strictly convex. Since $\Gfun$ is
uniquely determined within the class of differentiable functions under
these conditions, the resulting separating menu is unique among all
constructions derived from differentiable $\Gfun$.

When the principal controls all contract parameters, she can design
menus for any test satisfying the non-trivial power
assumption~\eqref{eq:nontrivial_power} and for all agent types on
$[0,1]$. This flexibility is powerful---it allows full elicitation of
private beliefs---but it also obscures how elicitation depends on the
structure of the hypothesis test itself, since rewards and costs can
always be adjusted to make things work. By contrast, when the reward
is fixed, the dependence on the test becomes much clearer: the power
function directly determines both the level of information rent and
the extent of elicitable types.

\paragraph{Elicitation depends on test structure:} Recall
that $\Gfun(\priorNull)$ represents the information rent allocated to
type $\priorNull$. Under the fixed-reward
construction~\eqref{eq:constR_G}, $\Gfun$ is pinned down directly by
the power function $\altAppSimple{\cdot}$, so that any change in the
test translates immediately into a change in information rent. For
example, if the principal employs a more powerful test $\tilde
\beta_1(\threshold)$ with $\tilde \beta_1(\threshold) >
\altAppSimple{\threshold}$ for all $\threshold \in [0,1]$, the
information rent increases strictly: greater statistical power raises
agents’ probabilities of receiving approvals and hence their utilities
under the base contract, which in turn requires higher rents to
maintain incentive compatibility. While information rent is always
driven by incentive compatibility, the fixed-reward setting makes its
dependence on the structure of the statistical test fully
transparent. The structural condition~\eqref{EqnPowerCondB} also
highlights the limits of elicitation under fixed rewards. Given a
concave, differentiable power function, the inequality
$\beta_1'(\threshold) > 1$ holds only on an interval $[0,\bar
  \threshold]$ (see \Cref{FigPSRConst}(a) for an illustration). This
places an upper bound $\bar \threshold$ on type-optimal thresholds in
any separating menu. Since $\fthreshold{\priorNull}$ is decreasing in
$\priorNull$, this implies that only agent types in
$[\bestPriorNull,\worstPriorNull]$ can be elicited, where
$\bestPriorNull$ is the type assigned threshold $\bar
\threshold$. Hence, although the principal intends to elicit all types
$\priorNull \in [0, \worstPriorNull]$, fixing the contract reward to
be constant imposes structural limits on which types can be elicited,
and these limits depend directly on the hypothesis test.

\paragraph{Financial cost of fixed-reward menus:} Compared to the menus constructed from $\Gfun$ in~\eqref{eq:varyingR_G}, the fixed-reward setting entails a non-trivial screening cost. With reward $\reward$ held constant, the curvature of $\Gfun$ is large, since $\Gfun''(\priorNull)$ cannot be tuned close to zero. As a result, the utility curve $\Gfun$ lies strictly above the linear baseline $\PsiBase(\priorNull)$ associated with the base contract, and the screening cost cannot be made arbitrarily small.
In some regulatory environments, however, the constant reward
$\reward$ is not borne by the principal---for example, when it
reflects the average market profits accruing to successful firms. In
such cases, the principal can implement the menu by charging
$\cost_{\worstPriorNull}$ to the worst type and requiring better types
to pay a surcharge $\cost_{\priorNull} - \cost_{\worstPriorNull}$ in
exchange for looser thresholds $\threshold_{\priorNull}$. Under this
arrangement, the principal earns strictly positive revenue from better
types, while maintaining incentive compatibility. This design is
Pareto-improving: the principal can screen agents by type and
implement type-optimal thresholds, while agents receive contracts
tailored to their beliefs and strictly higher utilities than under the
base contract. \\


\subsection{Numerical studies}
\label{secMain:opt_menu_plots}

Next, we present numerical results illustrating the convex functions
$\Gfun$ that underlie~\Cref{co:Gfun_varyingR,co:Gfun_constR}, defined
in equations~\eqref{eq:varyingR_G} and~\eqref{eq:constR_G},
respectively. We do so in the context of Gaussian mean testing:
consider a testing problem defined by the parameter space $\Theta =
\{0, \theta_1\}$, with a null value $\theta_0 = 0$ and a non-null
value $\theta_1 > 0$.  A given observation $Z \sim \mathcal{N}(\theta,
1)$ can be converted into a $p$-value via the transformation
$\evidence = 1 - \Phi(Z)$, where $\Phi$ is the standard normal
cumulative distribution function.  With this set-up, the null
rejection probability is $\beta_0(\threshold) = \threshold$ by
construction, while the power function under the alternative is given
by:
\begin{align*}
\altAppSimple{\threshold} &= 1 - \Phi(\Phi^{-1}(1-\threshold) -
\theta_1),
\end{align*}
where $\Phi^{-1}$ denotes the quantile function of the standard normal
distribution.

\begin{figure}[htbp]
    \centering \widgraph{0.5\textwidth}{\figdir/PSRMenu_varyingR}
    \caption{Principal's expected financial
      return~\eqref{DefPrincipalReturn} from offering a separating
      menu rather than the base contract
      $(\threshold_{\worstPriorNull}, \reward_{\worstPriorNull},
      \cost_{\worstPriorNull}) = (0.004, 100, 1.3)$ to each agent
      type. The separating menus are constructed using $\Gfun$ from
      equation~\eqref{eq:varyingR_G}, where $\eps{z} =
      \eta\cdot(1-z)^2$ for the choices $\eta \in \{0.01, 0.1, 0.5,
      1\}$.}
    \label{FigPSRVarying}
\end{figure}

First, consider a type-dependent reward with the function $\Gfun$ from
equation~\eqref{eq:varyingR_G}, chosen to maximize the principal's
expected financial return.  For illustration, we focus on an
alternative with $\theta_1 = 1$. We consider a worst agent type with
prior null probability $\worstPriorNull = 0.8$ and use type-optimal
thresholds~\eqref{EqnFDRThreshold} for FDR control. The false
discovery rate is constrained at level $0.25$, which implies an
optimal $p$-value threshold of $\threshold_{\worstPriorNull} = 0.004$
for this agent. The base contract fixes the reward at
$\reward_{\worstPriorNull} = 100$, with the cost
$\cost_{\worstPriorNull}$ calibrated so that the worst type obtains
zero utility under this contract.

\begin{figure}[h!]
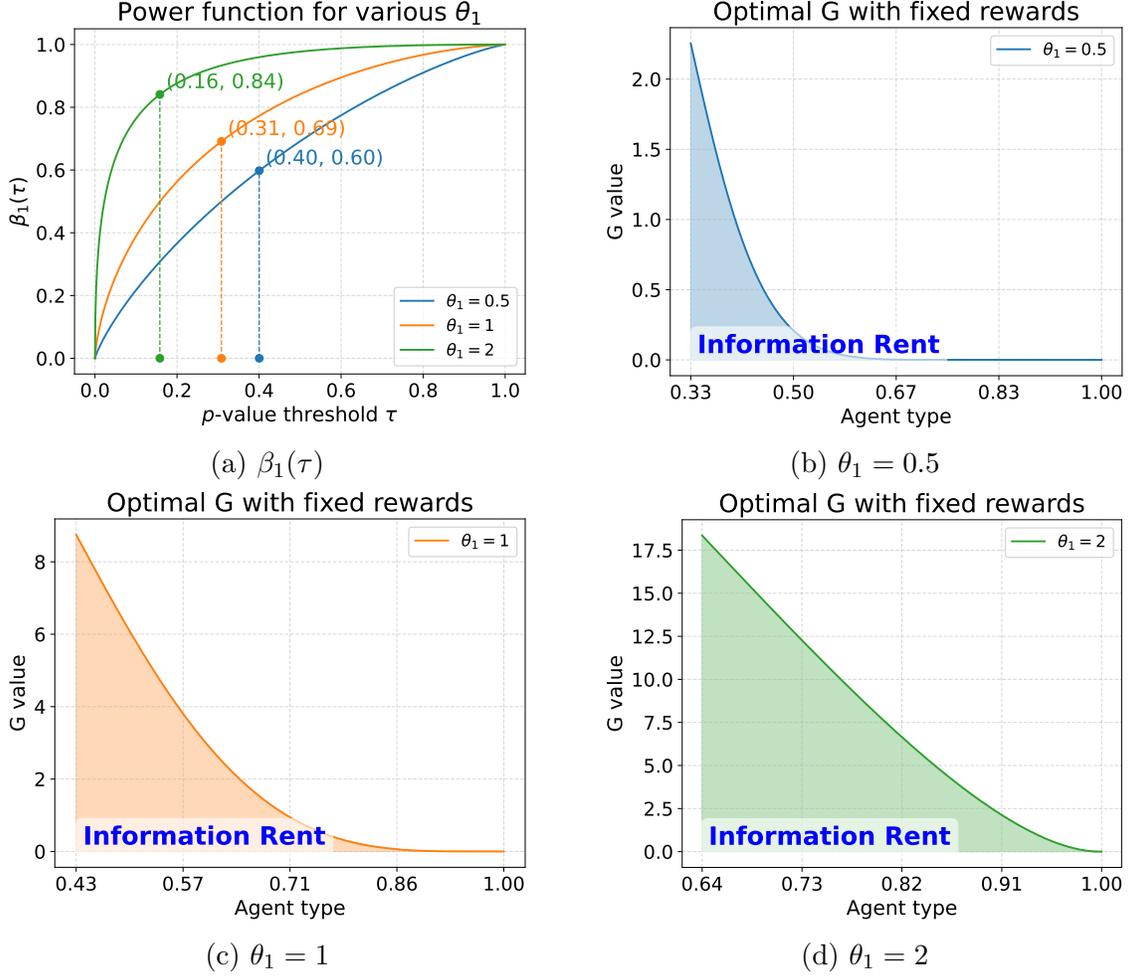

  \begin{center}
    \begin{tabular}{ccc}
    \widgraph{0.43\textwidth}{\figdir/PowerCurve_GMT} &&
    \widgraph{0.43\textwidth}{\figdir/PSRMenu_constR_mean_0.5} \\ (a)
    $\altAppSimple{\threshold}$ && (b) $\theta_1 = 0.5$
    \\ \widgraph{0.43\textwidth}{\figdir/PSRMenu_constR_mean_1} &&
    \widgraph{0.43\textwidth}{\figdir/PSRMenu_constR_mean_2} \\ (c)
    $\theta_1 = 1$ && (d) $\theta_1 = 2$
    \end{tabular}
    \caption{(a) Power functions $\altAppSimple{\threshold}$ under
      alternative hypotheses $\theta_1 \in \{0.5, 1, 2\}$, with the
      maximum threshold $\threshold$ satisfying
      condition~\eqref{EqnPowerCondB} marked. (b)--(d) Plots of the
      function $\Gfun$ from equation~\eqref{eq:constR_G} under the
      same alternatives. The shaded area indicates the total
      information rent for uniformly distributed agent types under the
      separating menu constructed from $\Gfun$. }
    \label{FigPSRConst}
  \end{center}
\end{figure}

As an illustration of the behavior of $\Gfun$, we consider the
perturbation function $\eps{z} \defn \eta\cdot(1-z)^2$ for some $\eta
> 0$; in \Cref{FigPSRVarying}, we plot the principal's expected
financial return~\eqref{DefPrincipalReturn}, relative to the base
contract, when using a separating menu constructed from $\Gfun$ for
the four choices $\eta \in \{0.01, 0.1, 0.5, 1\}$.  Here, the area
under each curve represents the principal's total financial loss for a
population of agents with uniformly distributed types. Observe that
setting $\eta$ close to zero results in smaller and smaller screening
cost.

We next turn to the case of a constant reward across all contracts. We
fix the reward at $100$ and again use type-optimal
thresholds~\eqref{EqnFDRThreshold} with an FDR constraint at level
$0.25$. As analyzed in~\eqref{eq:constR_G}, the function $\Gfun$ now
depends on the power function $\altAppSimple{\threshold}$. We consider
three different alternatives for the nonnull hypothesis: $\theta_1 \in
\{0.5, 1, 2\}$, with the null still given by $\theta_0 = 0$.

In~\Cref{FigPSRConst}(a), we plot the power functions corresponding to
these values of $\theta_1$, and mark the maximum $p$-value threshold
for which the condition $\beta_1^\prime(\threshold) > 1$ from
equation~\eqref{EqnPowerCondB} holds. As $\theta_1$ increases, the
power function becomes steeper, and the maximum threshold satisfying
this condition decreases. This restricts the range of agent types for
whom contracts with constant reward are incentive-compatible. As
discussed in \Cref{SecMain:opt_menu_varyingR}, this maximum threshold
imposes a lower bound on the agent types.  The effect is exhibited in
panels (b) through (d) of~\Cref{FigPSRConst}, where we plot the
corresponding functions $\Gfun$. For example, when $\theta_1 = 0.5$,
incentive-compatible contracts can be offered to agents with prior
null beliefs in the interval $[0.33, 1]$. As $\theta_1$ increases to
2, this range shifts toward agents with higher $\priorNull$, i.e.,
agents who are more likely to face the null hypothesis. This shift
occurs because when $\theta_1 = 2$, the maximum separating threshold
is approximately 0.16, which is too stringent to be optimal for agents
with lower $\priorNull$.  The shaded regions in panels (b)--(d)
represent the total information rent incurred when offering the menu
to a uniformly distributed population of agent types. As previously
discussed, we see that a more powerful test increases the information rent
required to screen better types.


\section{Discussion}
\label{SecDiscussion}

In this paper, we studied hypothesis testing over a heterogeneous
population of strategic agents with private information.  In this
setting, any single uniform test yields sub-optimal performance, due
to the underlying heterogeneity in agent types.  At the other extreme,
an oracle given \emph{a priori} access to agent types can construct an
optimal test for each type.  Our main result was to show that it is
possible for the principal to design a separating menu of
type-tailored tests, coupled with appropriately designed payoffs,
that induce agents to self-select according to their private
information, thereby implementing type-optimal thresholds and
achieving statistical efficiency. Strikingly, this improvement comes
at negligible additional cost relative to single-test designs,
demonstrating that information elicitation through incentive design
can lead to better statistical performance with little cost.

Beyond the technical results specific to this paper, our analysis also
offers some broader conceptual insights. First, it underscores the
importance of treating agents as strategic and heterogeneous:
conventional approaches that assume passive participants miss key
opportunities for improving statistical outcomes. Second, it
demonstrates how hypothesis tests themselves can serve as instruments
of information elicitation. By embedding contracts in statistical
design, the principal can shape agents’ utilities into a proper
scoring rule, ensuring incentive compatibility. More broadly, the
results further develop the connection between mechanism design and
statistical decision-making and highlight that design in strategic
settings requires jointly considering statistical objectives and agent
behavior.

There are a number of natural extensions to the current analysis.
First, we assumed that each agent’s private information is
one-dimensional, summarizing only their belief about type. In
practice, agents may hold richer, multi-dimensional information; for
example, a researcher’s prior knowledge across multiple outcomes or a
drug developer’s beliefs about several treatment effects.  Second, the
current analysis is predicated upon the principal knowing the test's
power function; it would be interesting to explore extensions
involving partial knowledge.  A final opportunity lies in the scope of
design and behavior considered. We focused on $p$-value thresholds and
binary participation, abstracting away other levers. In reality, a
regulator may be able to influence sample sizes, stopping rules, or
test statistics, while agents may strategically respond to these
choices. Extending the analysis to such richer action spaces could
capture a wider range of interactions.


\subsubsection*{Acknowledgements}
This work was partially funded by NSF-DMS-2413875 to SB and MJW, and the Cecil H. Green Chair and Ford Professorship to MJW.


\AtNextBibliography{\small} \printbibliography


\appendix

\renewcommand{\theequation}{A\arabic{equation}} 

\section{Separating menus for discrete types}
\label{SecApp:finite_type_menu}

In this section, we present a simpler procedure for constructing
separating menus when the population consists of a finite number
$\typeNum \geq 2$ of agent types. This setting can be viewed as a
special case of the general construction in~\Cref{th:psr_menu}, but
here we provide a direct derivation based on an iterative argument. We
first describe the procedure, then explain how it relates to the
convex-function approach, and finally illustrate its properties with
numerical examples.

Throughout, we assume that there are $\typeNum$ distinct agent types
with ordered prior null probabilities $\npriorNull{1} < \npriorNull{2}
< \dots < \npriorNull{\typeNum}$.  For menu design, it suffices that
the principal knows this support set $\supp(\typeDist) = \{
\npriorNull{t} \}_{t=1}^\typeNum$, without needing to know the
distribution $\typeDist$ over these types. Importantly, the type of
any individual agent remains private.

\subsection{Iterative procedure for menu construction}
\label{SecApp:finite_menu_algo}

We first set up the notation needed to describe our iterative procedure.
To avoid double subscripts, we index the contract menu by the
integer $t \in [\typeNum] \defn \{1, \dots, \typeNum \}$, so that any
menu takes the form $\ContractMenu = \big \{ (\threshold_t, \reward_t,
\cost_t) \mid t \in [\typeNum] \big \}$.  If an agent of type $\npriorNull{t}$ selects contract $t$, then the menu is incentive-compatible when $\SelFun(\npriorNull{t}) = \npriorNull{t}$ for all $t$.

We now introduce an iterative procedure for designing a contract menu
$\ContractMenu$. For $t \in [\typeNum]$,
define the scalars $\Delta_t \defn \altAppSimple{\threshold_t} -
\nullAppSimple{\threshold_t}$, along with the intervals $\Interval_t
\defn [\intleft_t, \intright_t]$ with endpoints
\begin{align*}
  \intleft_t & \defn \npriorNull{t} \big[\reward_t \Delta_t -
  \reward_{t-1}\Delta_{t-1}\big] +  \big[
  \reward_{t-1} \altAppSimple{\threshold_{t-1}} - \reward_t
  \altAppSimple{\threshold_t} + \cost_t \big],
\quad \mbox{and} \\
\intright_t & \defn \npriorNull{t-1} \big[ \reward_t \Delta_t -
  \reward_{t-1}\Delta_{t-1} \big] + \big[ \reward_{t-1}
  \altAppSimple{\threshold_{t-1}} - \reward_t
  \altAppSimple{\threshold_t} + \cost_t \big].
\end{align*}
The procedure takes as input the pair $(\reward_\typeNum,
\cost_\typeNum)$, and a sequence $\{ \eps{t} \}_{t=2}^\typeNum$ of
strictly positive scalars, and constructs the menu
$\{(\threshold_t, \reward_t, \cost_t \}_{t=1}^\typeNum$ as follows:
\begin{subequations} 
\begin{enumerate}
\item[(1)] For each $t \in [\typeNum]$, compute the type-optimal threshold
  $\threshold_{\npriorNull{t}}$ and
  set $\threshold_t \equiv \threshold_{\npriorNull{t}}$.
\item[(2)] In a backward recursion $t = \typeNum, \typeNum-1, \ldots,
  2$:
\begin{align}
\label{eq:finite_menu_reward}
{\bf{Update:}} \quad & \reward_{t-1} = \reward_{t}
\frac{\Delta_t}{\Delta_{t-1}} + \eps{t}, \\
\label{eq:finite_menu_cost}        
{\bf{Choose:}} \quad & \cost_{t-1} \in \big[\intleft_t, \intright_t
  \big].
\end{align}
\end{enumerate}

The following proposition certifies when this procedure correctly
computes a separating menu. In stating this result, it is helpful to
adopt the notation $\utilfunc(\priorNull; \threshold, \reward, \cost)$
to track the explicit dependence of the utility function~\eqref{EqnOptInUtility} on the triple
$(\threshold, \reward, \cost)$.
\begin{proposition}
\label{prop:finite_menu}
Assume that initial reward-cost pair $(\reward_\typeNum,
\cost_\typeNum)$ is chosen such that the participation constraint for
agent type $\npriorNull{\typeNum}$ holds---viz.
\begin{align}
\label{eq:finite_menu_participate}
    \utilfunc(\npriorNull{\typeNum}; \threshold_\typeNum,
    \reward_\typeNum, \cost_\typeNum) \geq 0.
\end{align}
Then for any strictly positive sequence $\{\eps{t} \}_{t=2}^\typeNum$,
the iterative procedure returns a contract menu $\ContractMenu =
\{(\threshold_t, \reward_t, \cost_t)\}_{t=1}^\typeNum$ that is
separating.
\end{proposition}
\end{subequations}
\noindent See~\Cref{proof:finite_menu} for the proof of this claim. \\

This iterative procedure can be understood as a discrete analogue of
the general convex-function construction in~\Cref{th:psr_menu}. In the
finite-type case, the separating menu induces a discrete convex
function $\mathcal{G}$.  Specifically,
step~\eqref{eq:finite_menu_reward} ensures $-\reward_{t-1}\Delta_{t-1}
< -\reward_t\Delta_t$, so that the slope of the utility function
strictly increases with type. Since the slope of the utility under
contract $(\threshold_\priorNull, \reward_\priorNull,
\cost_\priorNull)$ coincides with the subgradient of $\mathcal{G}$ at
$\priorNull$, step~\eqref{eq:finite_menu_reward} guarantees that
$\gstar{\npriorNull{t-1}} < \gstar{\npriorNull{t}}$, yielding a
sequence of increasing subgradients. Consequently, the
truthful-reporting utility becomes convex in agent
type. Step~\eqref{eq:finite_menu_cost} then chooses the intercept to
make the utility coincide with the supporting hyperplane of
$\mathcal{G}$ at each type. In this way, the procedure constructs a
discrete convex function satisfying property~\eqref{eq:Gfun_ic},
ensuring incentive
compatibility. Condition~\eqref{eq:finite_menu_participate} plays the
role of~\eqref{eq:Gfun_participation}, guaranteeing participation for
the worst type, and hence for all others. Finally, just as in the
general construction, separating menus are not unique. The slack
parameters ${\eps{t}}$ can be chosen arbitrarily (so long as they are
positive), and whenever the interval $[\intleft_t,\intright_t]$ is
non-degenerate, the cost $\cost_t$ can be varied. Thus there exists an
infinite family of separating menus. In the next subsection, we
illustrate how these choices affect the menus in numerical
simulations.

\subsection{Illustration with Gaussian mean testing}
\label{SecMain:finite_menu_plots}

In this section, we illustrate the separating menu constructed from procedure~\eqref{eq:finite_menu_reward}--\eqref{eq:finite_menu_cost} using the Gaussian mean testing setup from~\Cref{secMain:opt_menu_plots}. Consider a binary test with null $\theta_0 = 0$ and $\theta_1 = 1$. We
focus on a population consisting of $\typeNum = 5$ agent types, with
prior null probabilities $\priorNull$ ranging over the set
\mbox{$\supp(\typeDist) = \{ 0.3, 0.4, 0.5, 0.6, 0.7\}$.} With a
target false discovery rate $\fdrLevel = 0.25$, the
type-optimal $p$-value thresholds computed from~\eqref{EqnFDRThreshold} are $\{0.74, 0.38, 0.18, 0.07, 0.02\}$.  To ensure
that the participation constraint~\eqref{eq:finite_menu_participate}
is satisfied for agent $\typeNum = 5$ with $\priorNull = 0.7$, we set the cost-reward pair $(\cost_5, \reward_5) = (5,100)$.

To visualize the resulting menu, we plot the maximum expected utility attainable by each type $\priorNull \in [0,1]$, along with the contract that delivers this utility. Since the menu is separating, the resulting utility envelope coincides with a piecewise convex function. Each linear segment corresponds to the utility function~\eqref{EqnOptInUtility} induced by a specific contract and the range of types it attracts. The five designed types are highlighted in the plots, showing the specific contract each selects.

\Cref{FigFiniteMenu} makes explicit how the discrete construction connects to the general convex characterization. In one direction, the separating menu induces a discrete convex function $\mathcal{G}(\priorNull)$, represented by the five dots marking truthful-reporting utilities. The slope of the supporting line segment through each dot coincides with the subgradient of $\mathcal{G}$ at that type, which determines the utility available to all other types from the same contract. Convexity ensures that each type maximizes utility by reporting truthfully. Conversely, starting from this discrete convex function---or from the piecewise convex envelope---it is possible to reconstruct the separating menu, establishing the correspondence between discrete and general constructions.

\begin{figure}[ht]
  \begin{center}
    \begin{tabular}{ccc}
    \widgraph{0.45\textwidth}{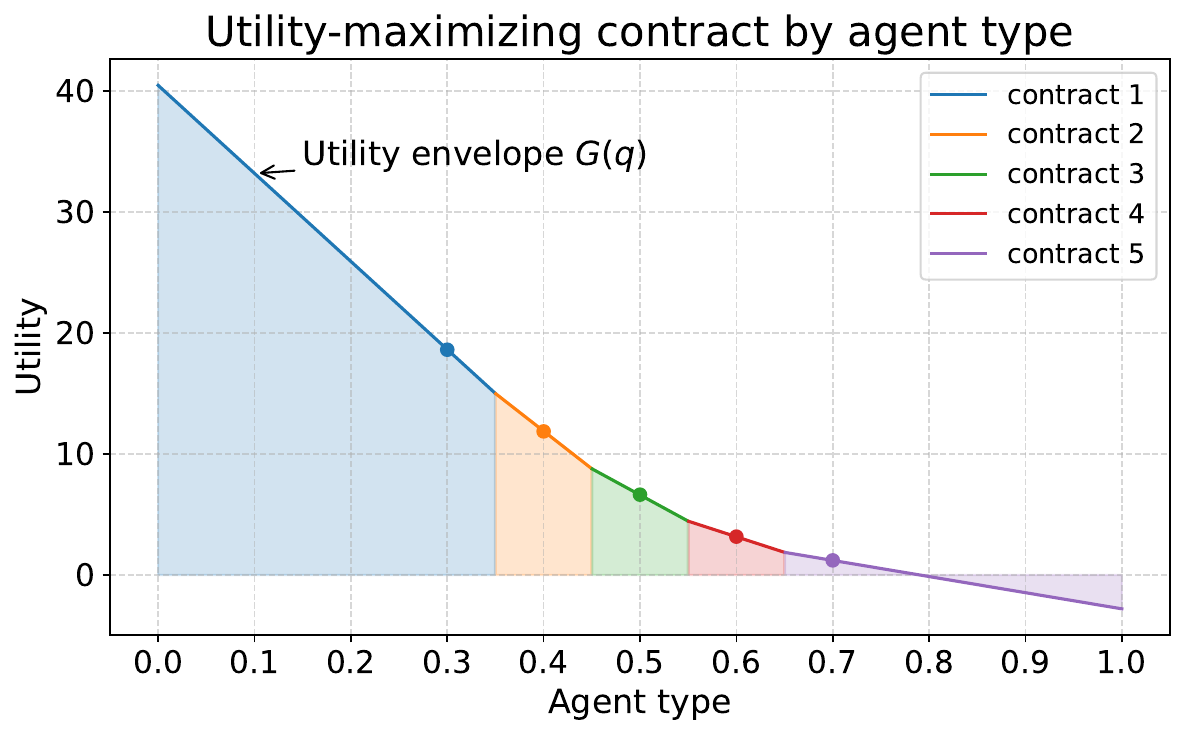}
      &&
      \widgraph{0.45\textwidth}{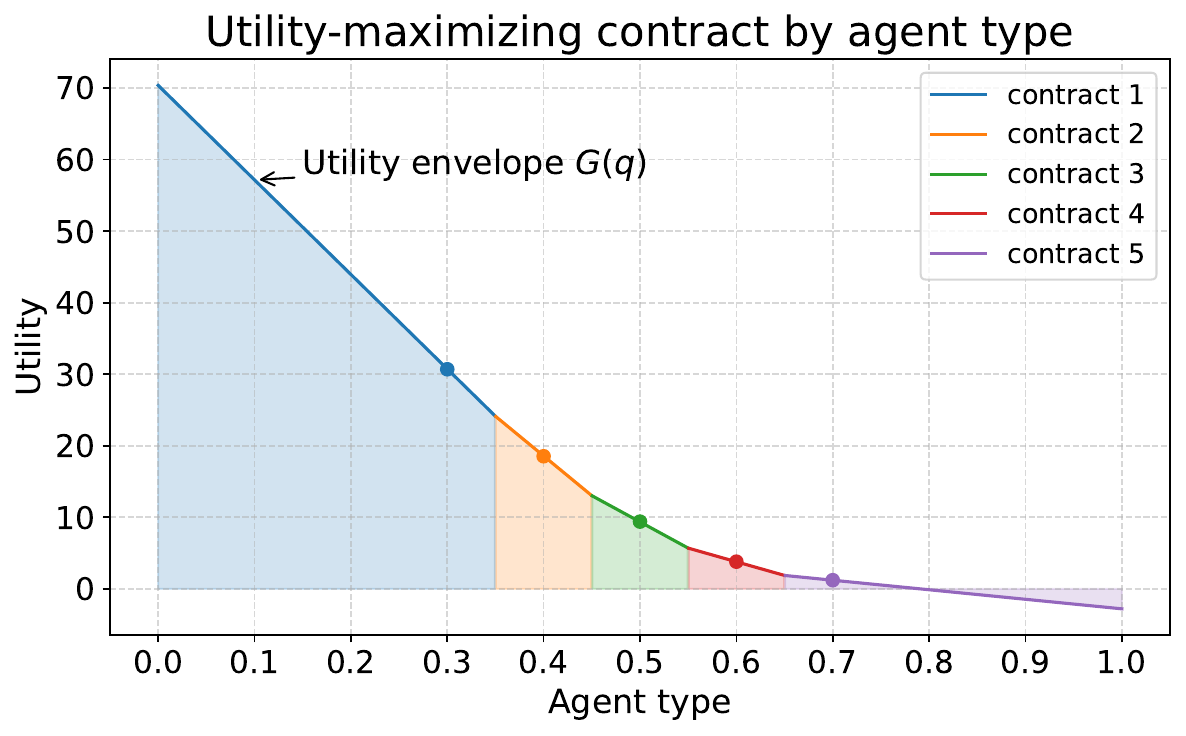} \\ (a) $\eps{t} =50$ && (b) $\eps{t} = 100$
    \\ \widgraph{0.45\textwidth}{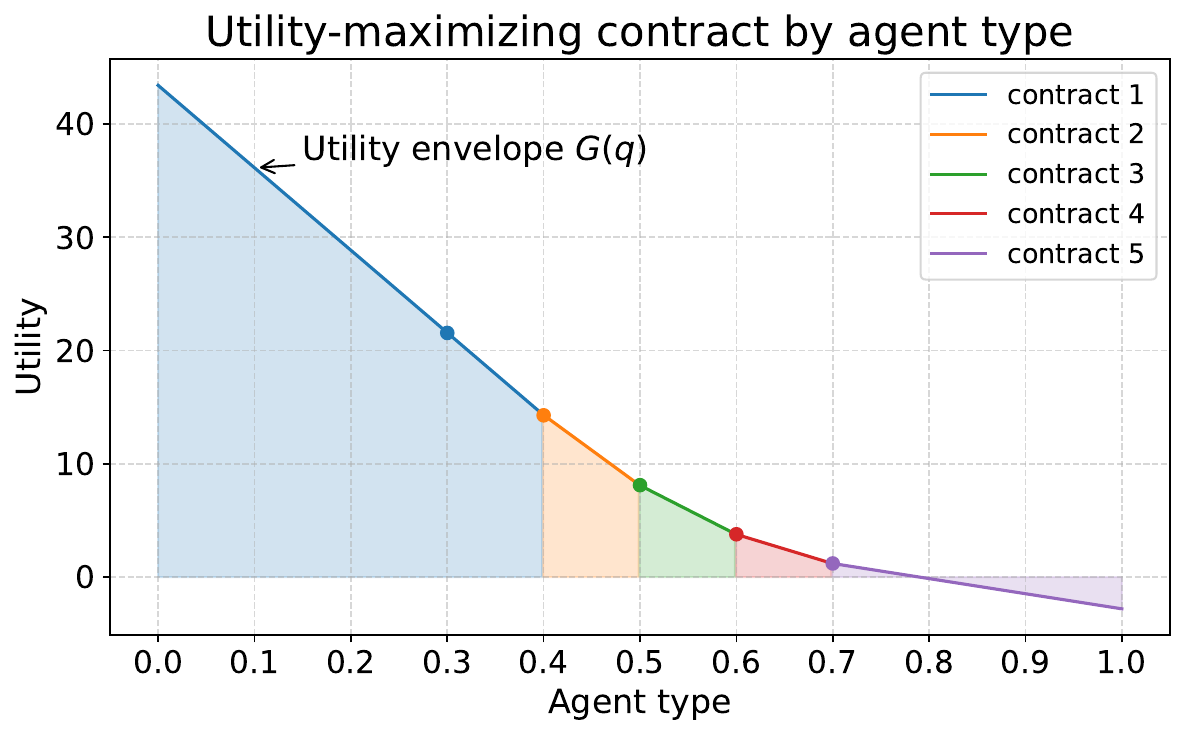}
      &&
      \widgraph{0.45\textwidth}{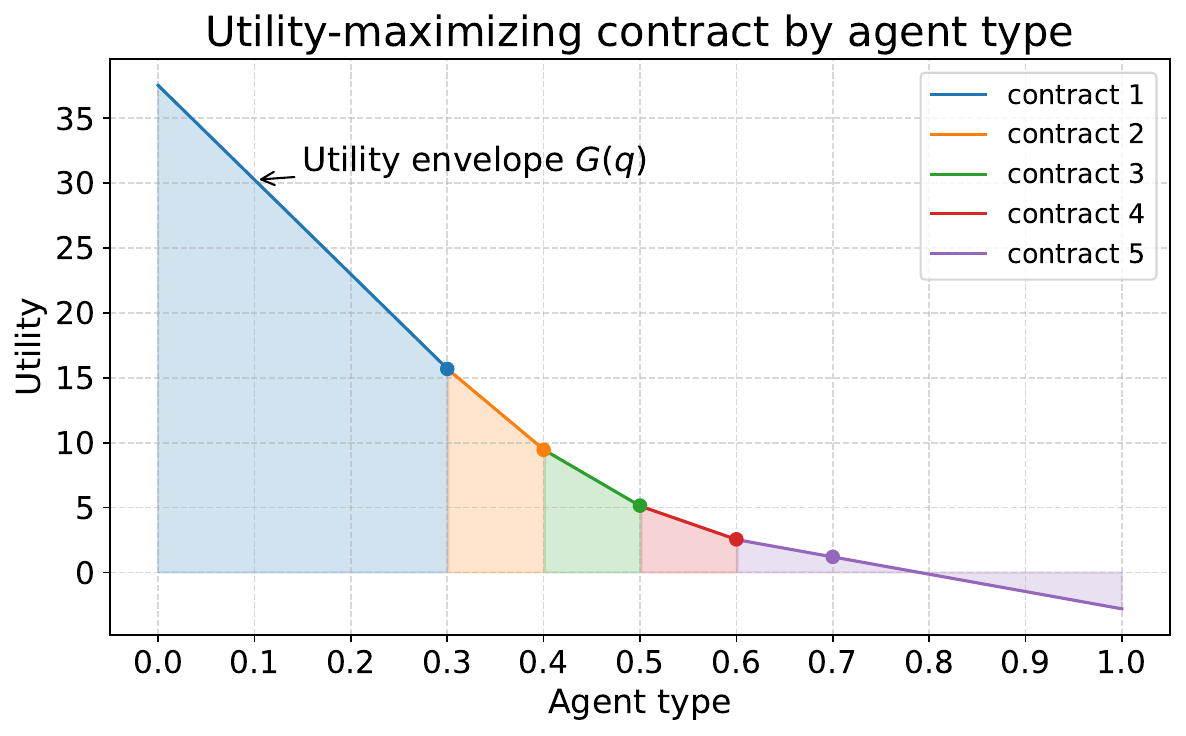} \\ (c)
    $\cost_{t-1}$ close to the lower bound $\intleft_t$ && (d) $\cost_{t-1}$
      close to the upper bound $\intright_t$
    \end{tabular}
    \caption{Plots of maximum utility~\eqref{EqnOptInUtility} achieved
      by each type $\priorNull \in [0,1]$ under a separating menu
      constructed for types $\priorNull \in \{ 0.3, 0.4, 0.5, 0.6,
      0.7\}$. The dots mark the utilities of the designed types under
      truthful reporting, forming a discrete convex function
      $\Gfun$. (a)-(b) Slack parameter $\eps{t} = 50$ and $\eps{t} =
      100$, respectively, with $\cost_{t-1} = \frac{1}{2}[\intleft_t +
        \intright_t]$. (c)-(d) Cost $\cost_{t-1} \approx \intleft_t$
      and $\cost_{t-1} \approx \intright_t$, respectively, with
      $\eps{t} = 50$.}
    \label{FigFiniteMenu}
  \end{center}
\end{figure}

\Cref{FigFiniteMenu} further illustrates how different choices of $\{\eps{t}
\}_{t=2}^\typeNum$ and
cost~\eqref{eq:finite_menu_cost} affect the resulting menu. In panel (a), we set
$\eps{t} = 50$ for all $t \in \{2, \dots, 5\}$; in panel (b), we
increase this value to $\eps{t} = 100$. In both cases, we set
$\cost_{t-1}$ to the midpoint of the admissible interval $[\intleft_t,
  \intright_t]$ from equation~\eqref{eq:finite_menu_cost}.  As we increase $\eps{t}$ from $50$ to $100$, we observe a
more pronounced increase in the slope of the utility functions from
contract 5 to contract 1. Intuitively, $\eps{t}$ controls the separation between types: larger values increase convexity of $\Gfun$ and strengthen incentive compatibility.

In panel (c), we
set $\cost_{t-1}$ to be the lower bound $\intleft_{t-1}$ of the
interval plus a small perturbation, while in panel (d), we set
$\cost_{t-1}$ to be the upper bound $\intright_{t-1}$ minus a small
perturbation, with both fixing $\eps{t} = 50$ for all $t \in \{2, 3, 4, 5\}$.  Comparing panels (a), (c), and (d), we see that changing
the choice of costs shifts the vertical intercept of the utility function induced by contract $t$. This shift then changes the location of the utility switch points---the values of $\priorNull$ at which an agent is indifferent between two adjacent contracts, i.e., where their expected utilities intersect. When $\cost_{t-1}$ is set near the lower bound of the interval, the switch point lies closer to type $\npriorNull{t}$, making contract $t$ nearly interchangeable with contract $t-1$ for that type. When $\cost_{t-1}$ is near the upper bound, the switch point moves toward $\npriorNull{t}$’s neighbor on the other side, making contract $t$ nearly interchangeable with contract $t+1$. Thus, the cost parameter effectively governs how the population is partitioned across adjacent contracts, while still preserving incentive compatibility.


\section{Separating menus under fixed costs}
\label{SecAppFixedCostMenu}
In this section, we show that constructing a separating menu in which all contracts share the same cost requires the power function to satisfy a form of convexity. We illustrate this in the simplest nontrivial setting with two agent types $\npriorNull{1} < \npriorNull{2}$ and corresponding type-optimal thresholds $\threshold_1 > \threshold_2$.

Assume the base contract $(\threshold_2, \reward_2, \cost_2)$ is chosen so that the participation constraint~\eqref{eq:finite_menu_participate} holds for type $\npriorNull{2}$. Following the procedure~\eqref{eq:finite_menu_reward}–\eqref{eq:finite_menu_cost}, the contract for type $\npriorNull{1}$ has reward
\begin{align}
\label{eq:fixCost_Reward}
    \reward_1 = \reward_2 \frac{\Delta_2}{\Delta_{1}} + \epsilon_1 \qquad \mbox{for some} \quad \epsilon_1 > 0,
\end{align}
where $\Delta_t = \altAppSimple{\threshold_t} - \nullAppSimple{\threshold_t}$. Now suppose both contracts are constrained to have the same cost, i.e., $\cost_1 = \cost_2 = \cost$. For incentive compatibility, type $\npriorNull{2}$ must strictly prefer its designated contract over the alternative:
\begin{align*}
   \utilfunc(\npriorNull{2};\threshold_2, \reward_2, \cost) > \utilfunc(\npriorNull{2};\threshold_1, \reward_1, \cost).
\end{align*}
Substituting utility function~\eqref{EqnOptInUtility} and expression~\eqref{eq:fixCost_Reward} yields
\begin{align*}
    \reward_2\big[\npriorNull{2} \nullAppSimple{\threshold_2} + (1-\npriorNull{2}) \altAppSimple{\threshold_2}\big] -\cost > \big[\reward_2 \frac{\Delta_2}{\Delta_{1}} + \epsilon_1 \big] \big[\npriorNull{2} \nullAppSimple{\threshold_1} + (1-\npriorNull{2}) \altAppSimple{\threshold_1}\big] -\cost.
\end{align*}
Since $\epsilon_1 > 0$, this inequality implies that
\begin{align*}
    \Delta_{1} \big[\npriorNull{2} \nullAppSimple{\threshold_2} + (1-\npriorNull{2}) \altAppSimple{\threshold_2}\big] > \Delta_2 \big[\npriorNull{2} \nullAppSimple{\threshold_1} + (1-\npriorNull{2}) \altAppSimple{\threshold_1}\big],
\end{align*}
which reduces to the following after simplification:
\begin{align*}
\nullAppSimple{\threshold_2}\altAppSimple{\threshold_1} > \nullAppSimple{\threshold_1}\altAppSimple{\threshold_2}
\end{align*}
With $\nullAppSimple{\threshold} =\threshold$, we conclude that for this fixed-cost menu to be incentive-compatible, the power function must satisfy
\begin{align}
\label{eq:fixCost_ratio}
    \frac{\altAppSimple{\threshold_1}}{\threshold_1} > \frac{\altAppSimple{\threshold_2}}{\threshold_2} \quad \mbox{for} \quad \threshold_1 > \threshold_2.
\end{align}

Inequality~\eqref{eq:fixCost_ratio} requires the ratio
$\altAppSimple{\threshold}/\threshold$ to be strictly increasing in
$\threshold$. This property corresponds to the power function
$\altAppSimple{\threshold}$ being convex (or at least displaying
convex-like curvature) on $[0,1]$.  However, in most statistical
settings, the power function of a test is concave in $\threshold$:
marginal increases in the significance level yield diminishing
improvements in power. Under concavity, we have the opposite
inequality---namely, $\frac{\altAppSimple{\threshold_1}}{\threshold_1}
\leq \frac{\altAppSimple{\threshold_2}}{\threshold_2}$---which
directly contradicts~\eqref{eq:fixCost_ratio}.

Thus, a separating menu with constant costs can only exist if the
power function is convex in $\threshold$, a property rarely satisfied
by common hypothesis tests. This explains why the constant-cost case
is generally infeasible in practice, and motivates our focus on the
constant-reward setting as the more realistic and broadly applicable
case.


\section{Sensitivity to model misspecification}
\label{Sec:sensitivity}

We have shown how to design separating menus that achieve statistical
efficiency under the assumption that the principal knows the power
function $\altAppSimple{\threshold}$ precisely. In practice, however,
it is often difficult to characterize the power function exactly. What
then can be said about statistical efficiency in such cases? When
$\altAppSimple{\threshold}$ is entirely unknown, \citet{shi2024sharp}
(Corollary 2) provides a worst-case bound on the false discovery rate
induced by any statistical contract. This represents a pessimistic
benchmark, reflecting the principal’s complete agnosticism about the
power function.

More realistically, the principal may possess partial information,
such as a set of plausible power functions---an assumption that
naturally arises in simple-versus-composite testing problems. Such
partial knowledge can support a more refined assessment of achievable
statistical efficiency. In this section, we focus on a principal whose
objective is to maximize TDR while controlling FDR at a prescribed
level $\fdrLevel$. We analyze the robustness of separating menus
constructed under this criterion to misspecification of the power
function, and in particular, how deviations from the assumed power
function may lead to FDR violations for certain agent types.

\paragraph{FDR gap.} 
Suppose the principal constructs a separating menu using a power
function $\altAppSimple{\threshold}$, which determines the contract
parameters $(\threshold_\reportNull, \reward_\reportNull,
\cost_\reportNull)$. The type-optimal thresholds
$\threshold_\reportNull$ are chosen according
to~\eqref{EqnFDRThreshold}. However, an agent of type $\priorNull$ may
hold beliefs aligned with a different power function
$\misspecified{\threshold}$.  Recall that under the assumed power
function $\altAppSimple{\threshold}$, the menu is designed to control
FDR for all reported types $\reportNull$:
\begin{align*}
 \frac{\reportNull \nullAppSimple{\threshold_\reportNull}
 }{\reportNull \nullAppSimple{\threshold_\reportNull} +
   (1-\reportNull)\altAppSimple{\threshold_\reportNull}} =
 \frac{\nullAppSimple{\threshold_\reportNull} }{
   \nullAppSimple{\threshold_\reportNull} +
   \frac{(1-\reportNull)}{\reportNull}
   \altAppSimple{\threshold_\reportNull} }\leq \fdrLevel.
\end{align*}
In contrast, if the true power function is $\misspecified{\threshold}$, the FDR incurred by an agent of type $\priorNull$ who reports $\reportNull$ becomes
\begin{align*}
 \frac{\priorNull \nullAppSimple{\threshold_\reportNull} }{\priorNull
   \nullAppSimple{\threshold_\reportNull} +
   (1-\priorNull)\misspecified{\threshold_\reportNull}} =
 \frac{\nullAppSimple{\threshold_\reportNull} }{
   \nullAppSimple{\threshold_\reportNull} +
   \frac{(1-\priorNull)}{\priorNull}
   \misspecified{\threshold_\reportNull}}.
\end{align*}
Hence, the sensitivity of FDR control to misspecification depends on
the difference between the denominators in these two expressions. We
define this difference as the \emph{FDR Gap}:
\begin{align}
\label{def:fdr_gap}
 \fdrGap \defn \frac{(1-\reportNull)}{\reportNull}
 \altAppSimple{\threshold_\reportNull} -
 \frac{(1-\priorNull)}{\priorNull}
 \misspecified{\threshold_\reportNull}.
\end{align}
If $\fdrGap \leq 0$, then even under misspecification, the FDR
constraint continues to hold. If $\fdrGap > 0$, the constraint is
violated, with the magnitude of the gap quantifying the degree of
violation. \\

To evaluate the FDR gap~\eqref{def:fdr_gap}, the principal must
determine which contract an agent of type $\priorNull$ selects. In
general, obtaining a closed-form expression for the selection
function~\eqref{EqnDefnSelFun} under misspecification is
challenging. To gain tractable insight, we analyze the special case of
a fixed-reward separating menu constructed from $\Gfun$
in~\eqref{eq:constR_G}. This setting allows us to characterize the
factors that drive FDR violations. Under such a fixed-reward menu, if
an agent with misspecified power function $\misspecified{\threshold}$
reports type $\reportNull$, the resulting FDR gap is
\begin{align}
\label{eq:fdr_gap_constR}
    \fdrGap = \frac{(1-\reportNull)}{\reportNull}
    \altAppSimple{\threshold_\reportNull} -
    \frac{1-\reportNull}{\reportNull \big[1 + \frac{\tilde
          \beta_1^\prime(\threshold_\reportNull) -
          \beta_1^\prime(\threshold_\reportNull)}{\reportNull \big(
          \beta_1^\prime(\threshold_\reportNull) - 1\big)} \big]}
    \misspecified{\threshold_\reportNull}.
\end{align}
See \Cref{SecProofFDRGap} for a proof. \Cref{eq:fdr_gap_constR}
highlights two main drivers of FDR sensitivity: (a) the difference in
local slopes of the power functions, i.e., $\tilde
\beta_1^\prime(\threshold_\reportNull) -
\beta_1^\prime(\threshold_\reportNull)$, and (b) the value of
$\misspecified{\threshold_\reportNull}$ itself.  If the principal can
restrict attention to a class of plausible power functions---such as
in simple-versus-composite hypothesis testing---this characterization
provides a concrete way to bound the extent of FDR violations arising
from misspecification.

To obtain a concrete illustration, we examine the sensitivity of FDR
control to misspecification in the setting of Gaussian mean
testing. Following the setup in~\Cref{secMain:opt_menu_plots}, we fix
the reward at $100$ across all contracts and impose an FDR constraint
of $0.25$. The separating menu is constructed under the assumption of
a binary test with null hypothesis $\theta_0 = 0$ and alternative
$\theta_1 = 1$. We consider $p$-value thresholds in the interval
$[0.001,0.31]$, which corresponds to a separating menu for agents with
$q \in [0.43,0.86]$.
\begin{figure}[ht]
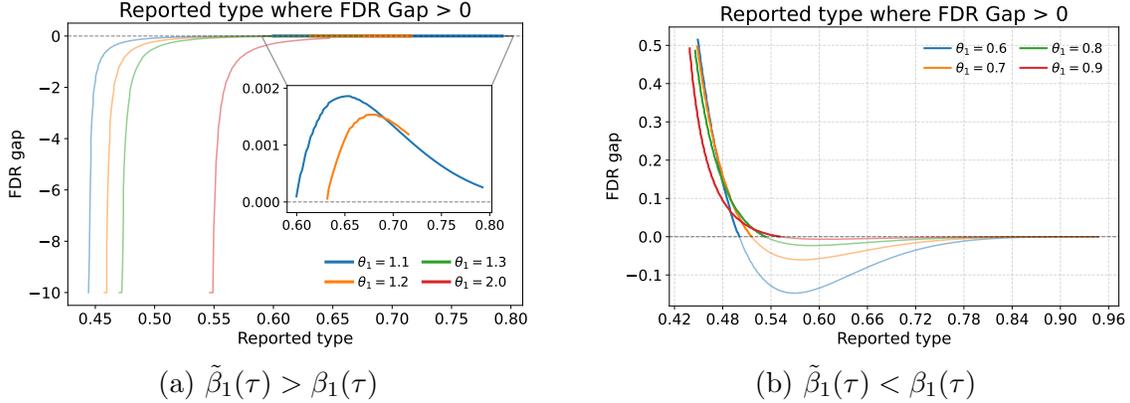

  \begin{center}
    \begin{tabular}{ccc}
      \widgraph{0.43\textwidth}{\figdir/Misspecified_overpower} &&
      \widgraph{0.43\textwidth}{\figdir/Misspecified_underpower}
      \\
      (a) $\misspecified{\threshold} > \altAppSimple{\threshold}$ &&
      (b) $\misspecified{\threshold} < \altAppSimple{\threshold}$
    \end{tabular}
    \caption{Plots of FDR gap~\eqref{eq:fdr_gap_constR} as a function
      of the reported type $\reportNull$, using power function
      $\misspecified{\threshold}$ under misspecified alternative
      hypotheses $\theta_1 \neq 1$. (a) $\theta_1 \in \{1.1, 1.2, 1.3,
      2\}$. (b) $\theta_1 \in \{0.6, 0.7, 0.8, 0.9\}$.}
    \label{FigMisPower}
  \end{center}
\end{figure}
In~\Cref{FigMisPower}, we plot the FDR gap~\eqref{eq:fdr_gap_constR}
for various misspecified power functions. Panel (a) shows the case
where $\misspecified{\threshold} > \altAppSimple{\threshold}$ with
$\theta_1 \in \{1.1, 1.2, 1.3, 2\}$, while panel (b) considers
$\misspecified{\threshold} < \altAppSimple{\threshold}$ with $\theta_1
\in \{0.6, 0.7, 0.8, 0.9\}$.

The results reveal a clear asymmetry. When the misspecified power
function is overpowered, FDR violations arise primarily for contracts
intended for agents with higher prior null. As shown in
~\Cref{FigMisPower}(a), for contracts associated with report types
$\reportNull < 0.6$, the FDR gap is negative and hence no violation
occurs. A narrow band of mild violations emerges for $0.6 <
\reportNull < 0.8$: for agents facing $\theta_1 = 1.1$ or $\theta_1 =
1.2$, the assigned thresholds $\threshold_\reportNull$ can yield FDR
violations, but the gap remains very small (upper bounded by $0.002$).
By contrast, when the power function is underpowered, violations
concentrate among contracts designed for agents with lower prior
null. In~\Cref{FigMisPower}(b), the FDR gap is positive mainly for
$\reportNull < 0.55$. In summary, the FDR violation in this instance
is negligible when the true $\theta_1$ is larger than the one used to
construct the menu, whereas it can be large in the opposite case. This
suggests that, in order to be robust, one should construct the menu
with respect to the weakest plausible alternative---that is, the one
yielding the lowest power.

\section{Proofs}
\label{SecProofs}

In this section, we collect the proofs of our results. 


\subsection{Proof of \texorpdfstring{\Cref{EqnOptInUtility}}{Equation (2c)}}
\label{SecProofOptInUtility}

For notational simplicity, we introduce the shorthand $W \defn
\WealthAfter{\reportNull}$, and note that $\utilfunc(\priorNull,
\reportNull) = \Exs[W]$ by definition.  A simple calculation yields
\begin{subequations}
\begin{align}
\Exs[W] & = \priorNull \Exs[W \mid \theta \in \Theta_0] + (1 -
\priorNull) \Exs[W \mid \theta \in \Theta_1] \notag \\
\label{EqnInitialDecomp}  
& = \priorNull \Big \{ \Exs[W \mid \theta \in \Theta_0] - \Exs[W \mid
  \theta \in \Theta_1] \Big \} + \Exs[W \mid \theta \in \Theta_1].
\end{align}
Given the definition~\eqref{eq:null_alt_app} of the type I error
$\nullAppSimple{\threshold}$, we can write
\begin{align}
\Exs[W \mid \theta \in \Theta_0] & =
\nullAppSimple{\threshold_\reportNull} \Exs[W \mid \theta \in
  \Theta_0, \evidence \leq \threshold_\reportNull] + \big(1 -
\nullAppSimple{\threshold_\reportNull} \big) \; \Exs[W \mid \theta \in
  \Theta_0, \evidence > \threshold_\reportNull] \nonumber \\
& = \nullAppSimple{\threshold_\reportNull} \big[\reward_\reportNull -
  \cost_\reportNull \big] - \big(1 -
\nullAppSimple{\threshold_\reportNull} \big) \cost_\reportNull
\nonumber \\
 \label{eq:null_expectation}    
& = \nullAppSimple{\threshold_\reportNull} \reward_\reportNull -
 \cost_\reportNull.
\end{align}
Similarly, the conditional expectation $\Exs[W \mid \theta \in
  \Theta_1]$ can be expressed as
\begin{align}
\Exs[W \mid \theta \in \Theta_1] & =
\altAppSimple{\threshold_\reportNull} \Exs[W \mid \theta \in \Theta_1,
  \evidence \leq \threshold_\reportNull] + \big(1 -
\altAppSimple{\threshold_\reportNull} \big) \; \Exs[W \mid \theta \in
  \Theta_1, \evidence > \threshold_\reportNull] \nonumber \\
 \label{eq:nonnull_expectation}
& = \altAppSimple{\threshold_\reportNull} \reward_\reportNull -
\cost_\reportNull.
\end{align}
\end{subequations}
Combining equations~\eqref{eq:null_expectation}
and~\eqref{eq:nonnull_expectation} with the initial
decomposition~\eqref{EqnInitialDecomp} yields
\begin{align*}
\utilfunc(\priorNull, \reportNull) = \E[W] = \priorNull
\reward_\reportNull \big[ \nullAppSimple{\threshold_\reportNull} -
  \altAppSimple{\threshold_\reportNull} \big] + \big[
  \reward_\reportNull \altAppSimple{\threshold_\reportNull} -
  \cost_\reportNull \big],
\end{align*}
as claimed.


\subsection{Proof of monotonicity of \texorpdfstring{\Cref{EqnFDRThreshold}}{Equation (6)}}
\label{SecProofLemMonotone}

In this proof, we assume that $\threshold_\priorNull > 0$
as this is the only case relevant to menu construction.  Recall that
an agent of type $\priorNull$ who selects a $p$-value threshold
$\threshold$ incurs the FDR given by
\begin{align*}
\psi(\priorNull, \threshold) \defn \frac{\priorNull
  \nullAppSimple{\threshold}}{\priorNull \nullAppSimple{\threshold} +
  (1-\priorNull) \altAppSimple{\threshold}}.
\end{align*}
To establish that the optimal threshold function
$\threshold_{\priorNull}$ is decreasing in the agent type
$\priorNull$, it suffices to show that for any pair $\priorNull <
\tilde \priorNull$, we have $\threshold_{\priorNull} \geq
\threshold_{\tilde \priorNull}$, where
\begin{align*}
 \threshold_{\priorNull} = \sup \big\{\threshold \, \mid \,
 \psi(\priorNull, \threshold) \leq \fdrLevel \big\}.
\end{align*}
For agent type $\tilde \priorNull$, the optimal threshold
$\threshold_{\tilde \priorNull}$ satisfies that $\psi(\tilde
\priorNull, \threshold_{\tilde \priorNull}) \leq \fdrLevel$.  The key step, which we show later, is that the function $\psi(\priorNull, \threshold)$ is strictly
increasing in $\priorNull$: 
\begin{align}
\label{eq:fdr_increase_type}
    \psi(\priorNull, \threshold) < \psi(\tilde \priorNull, \threshold) \quad \mbox{for any $\threshold > 0$ and
      pair $\priorNull < \tilde \priorNull$.}
\end{align}
This claim implies that $\psi(\priorNull, \threshold_{\tilde
  \priorNull}) < \psi(\tilde \priorNull, \threshold_{\tilde
  \priorNull}) \leq \fdrLevel$. By the definition of
$\threshold_\priorNull$ as the largest threshold satisfying the FDR
constraint, we must have $\threshold_{\priorNull} \geq
\threshold_{\tilde \priorNull}$.

\paragraph{Strict monotonicity:} When the power function $\altAppSimple{\threshold}$ is concave and differentiable, the function $\psi(\priorNull, \threshold)$ is also strictly increasing in $\threshold$, a fact we will prove later:
\begin{align}
\label{eq:fdr_increase_threshold}
    \psi(\priorNull, \threshold) < \psi( \priorNull, \tilde \threshold) \quad \mbox{for any $\priorNull > 0$ and
      pair $0 < \threshold < \tilde \threshold$.}
\end{align}
Since $\psi(\priorNull, \threshold_{\tilde \priorNull}) < \psi(\tilde \priorNull, \threshold_{\tilde \priorNull}) \leq \fdrLevel$, inequality~\eqref{eq:fdr_increase_threshold} implies that $\threshold_{\priorNull} > \threshold_{\tilde \priorNull}$. \\

\noindent We now present the proofs of the two auxiliary claims stated earlier.

\paragraph{Proof of claim\texorpdfstring{~\eqref{eq:fdr_increase_type}}{(A28)}:}

Taking the derivative of $\psi(\priorNull, \threshold)$ with respect to $\priorNull$ yields
\begin{align*}
\frac{d}{d\priorNull}\psi(\priorNull, \threshold) &=
\frac{\big[\priorNull \nullAppSimple{\threshold} + (1-\priorNull)
    \altAppSimple{\threshold}\big] \nullAppSimple{\threshold} -
  \priorNull \nullAppSimple{\threshold}\big[
    \nullAppSimple{\threshold} -
    \altAppSimple{\threshold}\big]}{\big[\priorNull
    \nullAppSimple{\threshold} + (1-\priorNull)
    \altAppSimple{\threshold}\big]^2} \\
& = \frac{\nullAppSimple{\threshold} \altAppSimple{\threshold}}{
  \big[\priorNull \nullAppSimple{\threshold} + (1-\priorNull)
    \altAppSimple{\threshold} \big]^2},
\end{align*}
which is strictly positive for $\threshold > 0$. Thus, $\psi$ is strictly increasing in $\priorNull$ for any fixed $\threshold > 0$ and we have established claim~\eqref{eq:fdr_increase_type}. 

\paragraph{Proof of claim\texorpdfstring{~\eqref{eq:fdr_increase_threshold}}{(A29)}:}

Differentiating the function $\psi(\priorNull, \threshold)$ with
respect to $\threshold$ yields
\begin{align*}
\frac{d}{d\threshold}\psi(\priorNull, \threshold) &=
\frac{\big[\priorNull \nullAppSimple{\threshold} + (1-\priorNull)
    \altAppSimple{\threshold}\big] \priorNull
  \beta_0^\prime(\threshold) - \priorNull
  \nullAppSimple{\threshold}\big[ \beta_0^\prime(\threshold) -
    \beta_1^\prime(\threshold)\big]}{\big[\priorNull
    \nullAppSimple{\threshold} + (1-\priorNull)
    \altAppSimple{\threshold}\big]^2} \\
& = \frac{\priorNull (1 - \priorNull) \big[ \beta_0^\prime(\threshold)
    \altAppSimple{\threshold} - \nullAppSimple{\threshold}
    \beta_1^\prime(\threshold)\big]}{\big[\priorNull
    \nullAppSimple{\threshold} + (1 - \priorNull)
    \altAppSimple{\threshold} \big]^2}.
\end{align*}
Since $\nullAppSimple{\threshold} = \threshold$, we have $\beta_0^\prime(\threshold) = 1$. Assumption~\eqref{eq:nontrivial_power} implies that $\altAppSimple{\threshold}$ is not linear function on $[0,1]$. The concavity of $\altAppSimple{\threshold}$ then implies that 
\begin{align*}
    \beta_1^\prime(\threshold) < \frac{\altAppSimple{\threshold} - \altAppSimple{0}}{\threshold} \quad \mbox{for all}\quad \threshold > 0,
\end{align*}
which is equivalent to \mbox{$\altAppSimple{\threshold} - \threshold
  \beta_1^\prime(\threshold) > 0$.}  Thus,
$\frac{d}{d\threshold}\psi(\priorNull, \threshold)$ is strictly
positive for $\threshold > 0$ and
claim~\eqref{eq:fdr_increase_threshold} follows.


\subsection{Proof of \texorpdfstring{\Cref{th:psr_menu}}{Theorem 1}}
\label{proof:psr_menu}

We split the proof into two parts. The first part establishes the
forward direction---namely, that any contract menu constructed
according to the specified procedure is separating---while the second
part proves the converse.

\subsubsection{Forward direction}

From the specification~\eqref{eq:psr_reward_cost} of the reward
$\reward_\reportNull$, 
the subgradient $\gstar{\reportNull}$ of the function $\Gfun$ satisfies
\begin{align}
\label{EqnMammoth}  
\gstar{\reportNull} = \reward_\reportNull \big [
  \nullAppSimple{\threshold_\reportNull} -
  \altAppSimple{\threshold_\reportNull} \big ].
\end{align}
The non-trivial power condition~\eqref{eq:nontrivial_power} combined
with condition~\eqref{eq:Gfun_ic} on $\Gfun$ ensure that
$\reward_\reportNull > 0$ for all $\reportNull \in
\supp(\typeDist)$.

Substituting the subgradient~\eqref{EqnMammoth} into expression~\eqref{eq:psr_reward_cost} for the cost
$\cost_\reportNull$ and rearranging yields 
\begin{align*}
\Gfun(\reportNull) & = \reward_\reportNull \Big[\reportNull
  \nullAppSimple{\threshold_\reportNull} + (1 - \reportNull)
  \altAppSimple{\threshold_\reportNull} \Big] - \cost_\reportNull.
\end{align*}
Thus, we see that the function value $\Gfun(\priorNull)$ is exactly
the utility~\eqref{EqnOptInUtility} of agent type $\priorNull$ under
truthful reporting.

In order to establish that the constructed menu is separating, we need
to show that
\begin{subequations}
\begin{align}
    \utilfunc(\priorNull; \priorNull) &> \utilfunc(\priorNull; \reportNull)
    && \mbox{for all} \quad \reportNull, \priorNull \in
    \supp(\typeDist) \; \mbox{and} \; \reportNull \neq
    \priorNull, \label{eq:thm1_proof_ic}\\ \utilfunc(\priorNull; \priorNull) &\geq 0 && \mbox{for all} \quad \priorNull \in
    \supp(\typeDist). \label{eq:thm1_proof_participation}
\end{align}
\end{subequations}
Conditions~\eqref{eq:thm1_proof_ic}
and~\eqref{eq:thm1_proof_participation} ensure that the menu is incentive-compatible and satisfies participation
constraints, respectively.

\paragraph{Proof of condition~\eqref{eq:thm1_proof_ic}.}

By property~\eqref{eq:Gfun_ic}, for all $\priorNull \in \supp(\typeDist)$, we have 
\begin{align*}
\Gfun(\priorNull) > \Gfun( \reportNull) + \gstar{\reportNull}(\priorNull - \reportNull) \qquad \mbox{for all
  $\reportNull \in \supp(\typeDist) \backslash \{\priorNull\}$.}
\end{align*}
Substituting expression~\eqref{EqnMammoth} and the relation
$\Gfun(\priorNull) = \utilfunc(\priorNull; \priorNull)$ yields
\begin{align*}
\utilfunc(\priorNull; \priorNull) > \utilfunc(\reportNull; \reportNull) +
\reward_{ \reportNull} \big[\nullAppSimple{\threshold_{ \reportNull}}
  - \altAppSimple{\threshold_{\reportNull}}\big] (\priorNull -
\reportNull) \qquad \mbox{for all $\reportNull \in \supp(\typeDist)
  \backslash \{\priorNull\}$.}
\end{align*}
Plugging in definition~\eqref{EqnOptInUtility} for
$\utilfunc(\reportNull; \reportNull)$ and rearranging, we see that the
right-hand side of this expression is equivalent to
\begin{align}
\label{eq:ic_equiv}
    \utilfunc(\reportNull; \reportNull) + \reward_{ \reportNull}
    \big[\nullAppSimple{\threshold_{ \reportNull}} -
      \altAppSimple{\threshold_{\reportNull}}\big] (\priorNull -
    \reportNull) = \reward_{\reportNull} \big[\priorNull
      \nullAppSimple{\threshold_{\reportNull}} + (1 - \priorNull)
      \altAppSimple{\threshold_{\reportNull}}\big] -
    \cost_{\reportNull} = \utilfunc(\priorNull; \reportNull).
\end{align}
Combining the pieces, we conclude that $\utilfunc(\priorNull;
\priorNull) > \utilfunc(\priorNull; \reportNull)$.



\paragraph{Proof of condition~\eqref{eq:thm1_proof_participation}.}
By property~\eqref{eq:Gfun_participation},
we are guaranteed to have
\begin{align*}
\Gfun(\worstPriorNull) = \utilfunc(\worstPriorNull; \worstPriorNull) \geq
0 \qquad \mbox{where $\worstPriorNull = \sup \{q \mid q \in
  \supp(\typeDist)\}$.}
\end{align*}
Therefore, condition~\eqref{eq:thm1_proof_participation} holds for
$\worstPriorNull$. Condition~\eqref{eq:thm1_proof_ic} implies that
$\Gfun(\priorNull) = \utilfunc(\priorNull; \priorNull) >
\utilfunc(\priorNull;\worstPriorNull)$ for all $\priorNull <
\worstPriorNull$. Based on the discussion
around~\Cref{EqnOptInUtility}, the function $\priorNull \mapsto
\utilfunc(\priorNull; \reportNull)$ is strictly decreasing, which
implies that $\utilfunc(\priorNull; \worstPriorNull) >
\utilfunc(\worstPriorNull; \worstPriorNull)$ \mbox{for all $\priorNull
  < \worstPriorNull$.}  Chaining these inequalities yields that
$\Gfun(\priorNull) > \Gfun(\worstPriorNull) \geq 0$ for all
$\priorNull < \worstPriorNull$.

\subsubsection{The converse}

Consider a separating menu $(\threshold_\reportNull,
\reward_\reportNull, \cost_\reportNull)$ indexed by agent types
$\priorNull \in \supp(\typeDist)$.  We now show that it can be
associated with a function $\Gfun$ that satisfies
conditions~\eqref{eq:Gfun_ic} and~\eqref{eq:Gfun_participation}.  In
particular, recalling the utility function~\eqref{EqnOptInUtility}, we
define $\Gfun(\priorNull) \defn \utilfunc(\priorNull;\priorNull)$,
which is the utility attained by each agent type under truthful
reporting. Moreover, we define the scalars
\begin{align}
\label{EqnGstarScalar}
\gstar{\priorNull} \defn \reward_\priorNull
\big[\nullAppSimple{\threshold_\priorNull} -
  \altAppSimple{\threshold_\priorNull}\big] \quad \mbox{for each
  $\priorNull \in \supp(\typeDist)$.}
\end{align}

\paragraph{Proof of property~\eqref{eq:Gfun_ic}.} Since the menu is
incentive-compatible, by condition~\eqref{eq:thm1_proof_ic}, we know
\begin{align*}
  \utilfunc(\priorNull; \priorNull) > \utilfunc(\priorNull; \reportNull)
  \quad \mbox{for all} \quad \reportNull, \priorNull \in
  \supp(\typeDist) \; \mbox{and} \; \reportNull \neq \priorNull.
\end{align*}
The equivalent expression in equation~\eqref{eq:ic_equiv} implies that
\begin{align*}
  \utilfunc(\priorNull; \priorNull) > \utilfunc(\reportNull; \reportNull)
  + \reward_{ \reportNull} \big[\nullAppSimple{\threshold_{
        \reportNull}} - \altAppSimple{\threshold_{\reportNull}}\big]
  (\priorNull - \reportNull).
\end{align*}
Equivalently, using the definition~\eqref{EqnGstarScalar}, we have
\begin{align*}
  \Gfun(\priorNull) > \Gfun(\reportNull) +
  \gstar{\reportNull}(\priorNull - \reportNull) \quad \mbox{for all}
  \quad \priorNull \in \supp(\typeDist).
\end{align*}
Moreover, definition~\eqref{EqnGstarScalar} combined with the
non-trivial power assumption~\eqref{eq:nontrivial_power} implies that
$\gstar{\reportNull} < 0$. Therefore, we conclude that
property~\eqref{eq:Gfun_ic} holds.

\paragraph{Proof of property~\eqref{eq:Gfun_participation}.}
Since the menu satisfies the participation
constraint~\eqref{eq:thm1_proof_participation}, we have
$\Gfun(\priorNull) =\utilfunc(\priorNull; \priorNull) \geq 0$ for all
$\priorNull \in
\supp(\typeDist)$. Property~\eqref{eq:Gfun_participation} thus holds
immediately.


\subsection{Proof of
\texorpdfstring{\Cref{co:Gfun_varyingR}}{Corollary 2}}
\label{proof:Gfun_varyingR}

In order to show that the menu is separating using~\Cref{th:psr_menu},
it suffices to show that the function $\Gfun$ from
equation~\eqref{eq:varyingR_G} satisfies properties~\eqref{eq:Gfun_ic}
and~\eqref{eq:Gfun_participation}.  By the non-trivial power
assumption~\eqref{eq:nontrivial_power}, we have
$\nullAppSimple{\threshold} - \altAppSimple{\threshold} < 0$.  From
the definition~\eqref{eq:varyingR_G}, we see that $\Gfun$ is
differentiable with derivative
\begin{align*}
  \underbrace{\Gfun'(\priorNull)}_{\equiv \gstar{\priorNull}} =
  \reward_{\worstPriorNull}\big[
    \nullAppSimple{\threshold_{\worstPriorNull}} -
    \altAppSimple{\threshold_{\worstPriorNull}}\big] (1 +
  \eps{\priorNull}).
\end{align*}
Since $\eps{\priorNull} \geq 0$, it follows that $\gstar{\priorNull} <
0$.  Moreover, because the mapping $\priorNull \mapsto
\eps{\priorNull}$ is differentiable, the second derivative of $\Gfun$
is given by
\begin{align*}
\Gfun''(\priorNull) = \frac{d}{d\priorNull} \Gfun'(\priorNull) =
\reward_{\worstPriorNull}\big[
  \nullAppSimple{\threshold_{\worstPriorNull}} -
  \altAppSimple{\threshold_{\worstPriorNull}}\big]
\specfunc^\prime(\priorNull).
\end{align*}
Since $\specfunc^\prime(\priorNull) < 0$, we have $\Gfun''(\priorNull)
> 0$, implying that $\Gfun$ is strictly convex on $[0,
  \worstPriorNull]$. Thus, we conclude that $\Gfun$ satisfies
property~\eqref{eq:Gfun_ic}. Finally, by
assumption~\eqref{eq:worst_participation}, we have
$\Gfun(\worstPriorNull) = 0$, verifying the second
condition~\eqref{eq:Gfun_participation}.

Next, we argue that the principal can make the screening
cost~\eqref{EqnDefnSeparation} arbitrarily small by adjusting the
values $\eps{z}$. It suffices to verify this for a single agent type
$\priorNull$, whose screening cost under~\eqref{eq:varyingR_G} is
given by
\begin{align*}
  \Bigg\{
  \reward_{\worstPriorNull}\big[
    \altAppSimple{\threshold_{\worstPriorNull}} -
    \nullAppSimple{\threshold_{\worstPriorNull}}\big]
  \int_{\priorNull}^{\worstPriorNull} \big[1 + \eps{z} \big]
  dz\Bigg\} - \Bigg \{ \priorNull \reward_{\worstPriorNull}
  \big[\nullAppSimple{\threshold_{\worstPriorNull}} -
    \altAppSimple{\threshold_{\worstPriorNull}}\big] +
  \big[\reward_{\worstPriorNull}\altAppSimple{\threshold_{\worstPriorNull}}
    - \cost_{\worstPriorNull} \big] \Bigg\}.
\end{align*}
Rearranging terms, we obtain
\begin{align*}
    -\bigg\{\reward_{\worstPriorNull}
    \big[\nullAppSimple{\threshold_{\worstPriorNull}} -
      \altAppSimple{\threshold_{\worstPriorNull}}\big] \bigg(
    \worstPriorNull + \int_{\priorNull}^{\worstPriorNull} \eps{z} dz
    \bigg) +
    \big[\reward_{\worstPriorNull}\altAppSimple{\threshold_{\worstPriorNull}}
      - \cost_{\worstPriorNull} \big] \bigg\}.
\end{align*}
As $\eps{z} \rightarrow 0^+$, the integral
$\int_{\priorNull}^{\worstPriorNull} \eps{z} dz$ converges to
zero. Consequently, the screening cost can be made arbitrarily close
to $-\utilfunc(\worstPriorNull; \threshold_{\worstPriorNull},
\reward_{\worstPriorNull}, \cost_{\worstPriorNull})$.  This proves the
claim, since $\utilfunc(\worstPriorNull; \threshold_{\worstPriorNull},
\reward_{\worstPriorNull}, \cost_{\worstPriorNull}) = 0$ by
assumption~\eqref{eq:worst_participation}.


\subsection{Proof of \texorpdfstring{\Cref{co:Gfun_constR}}{Corollary 3}}
\label{proof:Gfun_constR}

In order to show that the menu is separating using~\Cref{th:psr_menu},
it suffices to show that the function $\Gfun$ from
equation~\eqref{eq:constR_G} satisfies properties~\eqref{eq:Gfun_ic}
and~\eqref{eq:Gfun_participation}.  From its
definition~\eqref{eq:constR_G}, the function $\Gfun$ is differentiable
with derivative $\Gfun'(\priorNull) = \reward \big[
  \nullAppSimple{\threshold_\priorNull} -
  \altAppSimple{\threshold_\priorNull} \big]$.  By the non-trivial
power assumption~\eqref{eq:nontrivial_power} and the fact that
$\reward > 0$, it follows that $\Gfun'(\priorNull) \equiv
\gstar{\priorNull} < 0$.

Since $\altAppSimple{\threshold}$ is concave and differentiable, the
second derivative of $\Gfun$ is given by
\begin{align*}
\Gfun''(\priorNull) = \reward \big[
  \beta_0^\prime(\threshold_\priorNull) \threshold^\prime_\priorNull -
  \beta_1^\prime(\threshold_\priorNull) \threshold^\prime_\priorNull
  \big] \stackrel{(i)}{=} \reward \big[ 1 -
  \beta_1^\prime(\threshold_\priorNull) \big]
\threshold^\prime_\priorNull,
\end{align*}
where step (i) uses the identity $\nullAppSimple{\threshold} =
\threshold$ and the fact that $\beta_0^\prime(\threshold) = 1$.  By
assumption~\eqref{EqnStrictMonotone}, we have
$\threshold^\prime_\priorNull < 0$. Combined with
condition~\eqref{EqnPowerCondB}, this implies that
$\Gfun''(\priorNull)> 0$ and thus $\Gfun$ is strictly convex on
$[\bestPriorNull, \worstPriorNull]$. It follows that
property~\eqref{eq:Gfun_ic} holds.  Finally,
assumption~\eqref{eq:worst_participation} implies that
$\Gfun(\worstPriorNull) = 0$, verifying
property~\eqref{eq:Gfun_participation}.

Because $\Gfun$ is differentiable and its derivative is specified everywhere, the function is determined up to an additive constant. The condition $\Gfun(\worstPriorNull)=0$ removes this indeterminacy, yielding a unique function $\Gfun$. Through the correspondence in \Cref{th:psr_menu}, we conclude that the separating menu under fixed reward is uniquely determined.


\subsection{Proof of \texorpdfstring{\Cref{eq:fdr_gap_constR}}{Equation (A26)}}
\label{SecProofFDRGap}

Under a fixed-reward separating menu, the selection function for an agent of type $\priorNull$ with misspecified power function $\tilde\beta_1(\cdot)$ is
\begin{align}
 \label{eq:misspecified_report}  
 \SelFun(\priorNull) & = \arg \max_{\reportNull \in \reportSupp}
 \utilfunc(\priorNull; \reportNull) \; = \; \arg \max_{\reportNull
   \in \reportSupp} \reward \big[\priorNull
   \nullAppSimple{\threshold_\reportNull} + (1-\priorNull)
   \misspecified{\threshold_\reportNull} \big] - \cost_\reportNull,
\end{align}
where $\cost_\reportNull$ is given by~\eqref{eq:constR_cost}.
Differentiating the objective in~\eqref{eq:misspecified_report} with respect to $\reportNull$
and setting the derivative to zero yields
\begin{align*}
  \reward \threshold^\prime_{\reportNull} \big[\priorNull
      \beta_0^\prime(\threshold_\reportNull) + (1-\priorNull) \tilde
      \beta_1^\prime(\threshold_\reportNull) \big] = \reward
    \threshold^\prime_{\reportNull} \big[\reportNull
      \beta_0^\prime(\threshold_\reportNull) + (1-\reportNull)
      \beta_1^\prime(\threshold_\reportNull) \big].
\end{align*}
Cancelling $\reward \threshold^\prime_{\reportNull}$ from both sides and rearranging gives the condition that the reported type $\reportNull$ must satisfy:
\begin{align*}
 \priorNull \beta_0^\prime(\threshold_\reportNull) + (1-\priorNull)
 \tilde \beta_1^\prime(\threshold_\reportNull) = \reportNull
 \beta_0^\prime(\threshold_\reportNull) + (1-\reportNull)
 \beta_1^\prime(\threshold_\reportNull).
\end{align*}
Solving for $\priorNull$ in terms of $\reportNull$ gives $\priorNull
= \frac{\reportNull \beta_0^\prime(\threshold_\reportNull) +
  (1-\reportNull) \beta_1^\prime(\threshold_\reportNull) - \tilde
  \beta_1^\prime(\threshold_\reportNull)}{\beta_0^\prime(\threshold_\reportNull)
  - \beta_1^\prime(\threshold_\reportNull)}$.  Finally, substituting
this expression into the definition of the FDR
gap~\eqref{def:fdr_gap}, and using the fact that
$\beta_0^\prime(\threshold_\reportNull) = 1$, establishes
claim~\eqref{eq:fdr_gap_constR}.


\subsection{Proof of\texorpdfstring{~\Cref{prop:finite_menu}}{ Proposition 1}}
\label{proof:finite_menu}

Here we analyze the iterative procedure for constructing a separating
contract menu for finitely many agent types, as stated
in~\Cref{prop:finite_menu}.

\subsubsection{Main argument}

The recursive structure of
the menu construction ensures that local incentive compatibility between adjacent
types suffices to guarantee global truth-telling across all types.
Accordingly, the bulk of our effort is devoted to analyzing incentive
compatibility between adjacent types.  In particular, we begin by
stating two auxiliary lemmas that play a key role in the proof.

Our first lemma establishes conditions on rewards and costs that
ensure incentive compatibility between a pair of adjacent types.
Concretely, consider two adjacent agent types $\npriorNull{t}$ and
$\npriorNull{t-1}$ with $\npriorNull{t-1} < \npriorNull{t}$, and
 type-optimal thresholds $\threshold_t$ and
$\threshold_{t-1}$.  \Cref{lem:local_incentive_compatible}
specifies conditions on the reward-cost pairs that ensure local
incentive compatibility of the contracts $(\threshold_{t}, \reward_t,
\cost_t)$ and $(\threshold_{t-1}, \reward_{t-1}, \cost_{t-1})$.

\begin{lemma}
  \label{lem:local_incentive_compatible}
  The recursive construction~\eqref{eq:finite_menu_reward} and~\eqref{eq:finite_menu_cost}
  of reward-cost pairs ensures local incentive compatibility, meaning that
\begin{subequations}
\begin{align}
\label{eq:ic_for_good}
\utilfunc(\type{t-1}; \threshold_{t-1}, \reward_{t-1}, \cost_{t-1}) -
\utilfunc(\type{t-1}; \threshold_t, \reward_t, \cost_t) > 0, \\
\label{eq:ic_for_bad}    
\utilfunc(\type{t}; \threshold_t, \reward_t, \cost_t) -
\utilfunc(\type{t}; \threshold_{t-1}, \reward_{t-1}, \cost_{t-1}) > 0.
\end{align}
\end{subequations}
\end{lemma}
\noindent See~\Cref{SecProoflem:local_incentive_compatible} for the proof.\\

\noindent Our second lemma formalizes an important type of monotonic
preference structure in the contracts, ensuring that no agent prefers
a contract intended for a more distant type.
\begin{lemma}
\label{lem:ICOrdering}
The contracts in any menu constructed by the iterative procedure
satisfy the properties
\begin{subequations}
  \begin{align}
    \label{eq:IC_worseOrder}
    \utilfunc(\priorNull;\threshold_t, \reward_t, \cost_t) &>
    \utilfunc(\priorNull;\threshold_{t-1}, \reward_{t-1}, \cost_{t-1})
    \quad \mbox{for all} \quad \priorNull \geq \type{t}, \\
\label{eq:IC_betterOrder}    
\utilfunc(\priorNull;\threshold_t, \reward_t, \cost_t) &<
\utilfunc(\priorNull;\threshold_{t-1}, \reward_{t-1}, \cost_{t-1})
\quad \mbox{for all} \quad \priorNull \leq \type{t-1}.
\end{align} 
\end{subequations}
\end{lemma}
\noindent See~\Cref{SecProoflem:ICOrdering} for the proof. \\

With these two lemmas in place, we are now prepared to prove the
proposition itself. \Cref{lem:local_incentive_compatible} ensures that
any pair of adjacent contracts is locally incentive-compatible.  Our
first step is to leverage this property so as to establish
\emph{global} incentive compatibility, meaning that for each $t \in
     [\typeNum]$, we have
\begin{align}
\label{EqnGlobal}       
 \utilfunc(\type{t}; \threshold_{t}, \reward_{t}, \cost_{t}) >
 \utilfunc(\type{t}; \threshold_{s}, \reward_{s}, \cost_{s}) \quad
 \mbox{for all} \quad s \in [\typeNum] \backslash \{t\}.
\end{align}

To establish this property, first consider some $s \leq t$ and hence the contract $(\threshold_{s}, \reward_{s}, \cost_{s})$ intended for $\npriorNull{s} \leq \npriorNull{t}$. Equation~\eqref{eq:IC_worseOrder}
from~\Cref{lem:ICOrdering} implies that for agent type $ \npriorNull{t}$,
\begin{align*}
\utilfunc(\type{t}; \threshold_{s}, \reward_{s}, \cost_{s}) >
\utilfunc(\type{t}; \threshold_{s-1}, \reward_{s-1}, \cost_{s-1}).
\end{align*}
Chaining together relations of this type yields
\begin{align*}
    \utilfunc(\type{t}; \threshold_{t}, \reward_{t}, \cost_{t}) >
    \utilfunc(\type{t}; \threshold_{s}, \reward_{s}, \cost_{s}) \quad \mbox{for all} \quad s < t.
\end{align*}
Similarly, for any $s > t$, inequality~\eqref{eq:IC_betterOrder}
from~\Cref{lem:ICOrdering} implies that for agent type $\npriorNull{t}$,
\begin{align*}
\utilfunc(\type{t}; \threshold_{s-1}, \reward_{s-1}, \cost_{s-1}) >
\utilfunc(\type{t}; \threshold_{s}, \reward_{s}, \cost_{s}).
\end{align*}
Chaining these relations yields
\begin{align*}
\utilfunc(\type{t}; \threshold_{t}, \reward_{t}, \cost_{t}) >
\utilfunc(\type{t}; \threshold_{s}, \reward_{s}, \cost_{s}) \quad \mbox{for all} \quad s > t.
\end{align*}
Thus, the global incentive compatibility constraint is satisfied for
each agent type.

Now, we show that the constructed menu satisfies the participation
constraints, i.e., for each agent of type $\type{t}$, we have
$\utilfunc(\type{t}; \threshold_t, \reward_t, \cost_t) \geq 0$.  By
assumption, this inequality holds trivially for agent type
$\type{\numType}$. For any $t < \numType$, by the
incentive compatibility constraints, we have
\begin{align*}
\utilfunc(\type{t}; \threshold_t, \reward_t, \cost_t) >
\utilfunc(\type{t}; \threshold_\numType, \reward_\numType,
\cost_\numType).
\end{align*}
Since $\type{t} < \type{\numType}$ for all $t < \numType$ and the
function $\priorNull \mapsto \utilfunc(\priorNull;
\threshold_\numType, \reward_\numType, \cost_\numType)$ is monotonically
decreasing in $\priorNull$, we have
\begin{align*}
    \utilfunc(\type{t}; \threshold_\numType, \reward_\numType,
    \cost_\numType) > \utilfunc(\type{\numType}; \threshold_\numType,
    \reward_\numType, \cost_\numType) \geq 0.
\end{align*}
Consequently, we conclude that $\utilfunc(\type{t}; \threshold_t,
\reward_t, \cost_t) \geq 0$ for all $t < \numType$, as claimed.


\subsubsection{Proof of \texorpdfstring{\Cref{lem:local_incentive_compatible}}{Lemma 1}}
\label{SecProoflem:local_incentive_compatible}

Recalling the shorthand $\Delta_t \defn \altAppSimple{\threshold_t} -
\nullAppSimple{\threshold_t}$, the utility
function~\eqref{EqnOptInUtility} takes the form
\begin{align*}
\utilfunc(\priorNull; \threshold_t, \reward_t, \cost_t) & =
-\priorNull \reward_t \Delta_t +
\big[\reward_t\altAppSimple{\threshold_t} - \cost_t \big].
\end{align*}
We substitute this utility function into
equations~\eqref{eq:ic_for_good} and~\eqref{eq:ic_for_bad}; by
re-arranging, we find that the cost $\cost_{t-1}$ must satisfy the
upper and lower bounds
\begin{subequations}
  \begin{align*}
    \cost_{t-1} & < \intright_t \defn \type{t-1} \Big(\reward_t
    \Delta_t - \reward_{t-1} \Delta_{t-1} \Big) + \Big[\reward_{t-1}
      \altAppSimple{\threshold_{t-1}} - \reward_t
      \altAppSimple{\threshold_t} + \cost_t \Big], \quad \mbox{and} \\
    \cost_{t-1} & > \intleft_t \defn \type{t} \Big(\reward_t \Delta_t
    -\reward_{t-1} \Delta_{t-1}\Big) + \Big[\reward_{t-1}
      \altAppSimple{\threshold_{t-1}} -
      \reward_t\altAppSimple{\threshold_t} + \cost_t \Big].
\end{align*}
\end{subequations}
In order to guarantee the existence of such $\cost_{t-1}$, we require
\begin{align*}
\type{t-1} \Big(\reward_t\Delta_t -\reward_{t-1}\Delta_{t-1} \Big) >
\type{t} \Big(\reward_t\Delta_t -\reward_{t-1}\Delta_{t-1}\Big).
\end{align*}
Since $\type{t} > \type{t-1}$, this condition can be satisfied by
requiring that
\begin{align}
\label{EqnMina}
\reward_t \Delta_t - \reward_{t-1} \Delta_{t-1} < 0.
\end{align}
Recall that our non-trivial power assumption ensures that
$\altAppSimple{\threshold} > \nullAppSimple{\threshold}$ for all
$\threshold \in [0,1]$; as a consequence, we have $\Delta_t > 0$ and
$\Delta_{t-1} > 0$.  Re-arranging yields the equivalent inequality
$\reward_{t-1} > \reward_t \frac{\Delta_t}{\Delta_{t-1}}$, which is
satisfied by setting
\begin{align*}
  \reward_{t-1} = \reward_t \frac{\Delta_t}{\Delta_{t-1}} + \eps{t}
  \qquad \mbox{for some $\eps{t} > 0$.}
\end{align*}
This completes the proof.


\subsubsection{Proof of \texorpdfstring{\Cref{lem:ICOrdering}}{Lemma 2}}
\label{SecProoflem:ICOrdering} 
              
Recall that the expected utility function $\priorNull \mapsto
\utilfunc(\priorNull; \threshold_t, \reward_t, \cost_t)$ is linear in
$\priorNull$ and has a negative slope $-\reward_t\Delta_t$.
\Cref{lem:local_incentive_compatible} ensures that $
-\reward_t\Delta_t > -\reward_{t-1}\Delta_{t-1}$.  Using the linear
structure, we can write
\begin{align*}
\utilfunc(\priorNull;\threshold_t, \reward_t, \cost_t) &=
\utilfunc(\type{t};\threshold_t, \reward_t, \cost_t) - (\priorNull -
\type{t}) \reward_t \Delta_t \\
\utilfunc(\priorNull;\threshold_{t-1}, \reward_{t-1}, \cost_{t-1}) &=
\utilfunc(\type{t};\threshold_{t-1}, \reward_{t-1}, \cost_{t-1}) -
(\priorNull - \type{t}) \reward_{t-1} \Delta_{t-1}
\end{align*}
Now for any prior null probability $\priorNull \geq \type{t}$, we have
\begin{align*}
  -(\priorNull - \type{t}) \reward_t\Delta_t > -(\priorNull -
  \type{t}) \reward_{t-1} \Delta_{t-1}.
\end{align*}
Given incentive compatibility, i.e., $\utilfunc(\type{t};\threshold_t,
\reward_t, \cost_t) > \utilfunc(\type{t};\threshold_{t-1},
\reward_{t-1}, \cost_{t-1})$, we conclude that
\begin{align*}
\utilfunc(\priorNull;\threshold_t, \reward_t, \cost_t) >
\utilfunc(\priorNull;\threshold_{t-1}, \reward_{t-1}, \cost_{t-1})
\quad \mbox{for all} \quad \priorNull \geq \type{t}.
\end{align*}
The proof of \cref{eq:IC_betterOrder} is similar. Using the linear
structure, we can write
\begin{align*}
\utilfunc(\priorNull;\threshold_t, \reward_t, \cost_t) &=
\utilfunc(\type{t-1};\threshold_t, \reward_t, \cost_t) - (\priorNull -
\type{t-1}) \reward_t \Delta_t \\
\utilfunc(\priorNull;\threshold_{t-1}, \reward_{t-1}, \cost_{t-1}) &=
\utilfunc(\type{t-1};\threshold_{t-1}, \reward_{t-1}, \cost_{t-1}) -
(\priorNull - \type{t-1}) \reward_{t-1} \Delta_{t-1}
\end{align*}
For all $\priorNull \leq \type{t-1}$, we have $- (\priorNull -
\type{t-1}) \reward_{t} \Delta_{t} < - (\priorNull - \type{t-1})
\reward_{t-1} \Delta_{t-1}$.  Combining this inequality with
incentive compatibility yields
\begin{align*}
  \utilfunc(\priorNull;\threshold_t, \reward_t, \cost_t) <
  \utilfunc(\priorNull;\threshold_{t-1}, \reward_{t-1}, \cost_{t-1})
  \quad \mbox{for all} \quad \priorNull \leq \type{t-1},
\end{align*}
which completes the proof.


\end{document}